\documentclass{amcEA}
\usepackage{amsmath}
\usepackage{paralist}
\usepackage[misc]{ifsym}
\usepackage{epsfig} 
\usepackage{epstopdf} 
\usepackage[colorlinks=true]{hyperref}
\hypersetup{urlcolor=blue, citecolor=red}
\allowdisplaybreaks

\usepackage{booktabs,caption}
\usepackage[flushleft]{threeparttable}
\usepackage{amsfonts,amssymb,amsthm,enumitem}
\usepackage{adjustbox}
\usepackage{makecell}
\usepackage{url}            % simple URL typesetting
\usepackage{comment}   %allows for multi-line comments
\usepackage{bm}
\usepackage{arydshln}

\def\currentvolume{X}
 \def\currentissue{X}
  \def\currentyear{200X}
   \def\currentmonth{XX}
    \def\ppages{X-XX}
     \def\DOI{10.3934/amc.2024007}

\newcommand{\F}{\mathbb{F}}
\newcommand{\E}{\mathop \mathbb{E}}
\newcommand{\N}{\mathbb{N}}
\newcommand{\R}{\mathbb{R}}

\newtheorem{theorem}{Theorem}[section]
\newtheorem{corollary}[theorem]{Corollary}

\newtheorem{lemma}[theorem]{Lemma}

\newtheorem{prop}[theorem]{Proposition}
\newtheorem{claim}[theorem]{Claim}
\newtheorem{fact}[theorem]{Fact}

\theoremstyle{definition}
\newtheorem{definition}[theorem]{Definition}

\title[A Criterion for Decoding on the BSC]
{A criterion for decoding on the binary symmetric channel} 
\author[Anup Rao and Oscar Sprumont]{}



\thanks{The first author was supported by NSF CCF-2131899. The second author was supported in part by NSERC PGSD3-545945-2020, NSF CCF-2131899, NSF CCF-1813135 and Anna Karlin's Bill and Melinda Gates Endowed Chair.}

\thanks{$^*$Corresponding author: Oscar Sprumont}

\begin{document}
\maketitle

\centerline{\scshape
Anup Rao$^{{\href{mailto:anuprao@cs.washington.edu}{\textrm{\Letter}}}}$
and Oscar Sprumont$^{{\href{mailto:osprum@cs.washington.edu}{\textrm{\Letter}}}*}$}

\medskip

{\footnotesize

 \centerline{School of Computer Science, University of Washington, USA}
}

\bigskip

\begin{abstract}
We present an approach to showing that a linear code is resilient to random errors. We use this approach to obtain decoding results for both transitive and doubly transitive codes. We give three kinds of results about linear codes in general, and transitive linear codes in particular.
\begin{enumerate}
\item We give a tight bound on the weight distribution of every transitive linear code $C \subseteq \F_2^N$:
$$\Pr_{c \in C}[\textnormal{wt}(c) = \alpha N] \leq 2^{-(1-h(\alpha)) \mathsf{dim}(C)}. $$
\item We give a criterion that certifies that a linear code $C$ can be decoded on the binary symmetric channel. Let $K_s(x)$ denote the Krawtchouk polynomial of degree $s$, and let $C^\perp$ denote the dual code of $C$. We show that bounds on $\E_{c \in C^{\perp}}[ K_{\epsilon N}(\textnormal{wt}(c))^2]$ imply that $C$ recovers from errors on the binary symmetric channel with parameter $\epsilon$. Weaker bounds can be used to obtain list-decoding results using similar methods. One consequence of our criterion is that whenever the weight distribution of $C^\perp$ is sufficiently close to the binomial distribution in some interval around $\frac{N}{2}$, $C$ is resilient to $\epsilon$-errors.
\item We combine known estimates for the Krawtchouk polynomials with our weight bound for transitive codes, and with known weight bounds for doubly transitive codes, to obtain list-decoding results for both these families of codes. In some regimes, our bounds for doubly transitive codes achieve the information-theoretic optimal trade-off between rate and list size.
\end{enumerate}
\end{abstract}

\section{Introduction}\label{intro}
In his seminal 1948 paper, Shannon laid out the bases of coding theory and introduced the concept of channel capacity, which is the maximal rate at which information can be transmitted over a communication channel \cite{shannon1948entropy}. The two channels that have received the most attention are the Binary Symmetric Channel (BSC), where each bit is independently flipped with some probability $\epsilon$, and the Binary Erasure Channel (BEC), where each bit is independently replaced by an erasure symbol with some probability $\epsilon$. Shannon's work initiated a decades-long search for explicit codes that can achieve high rates over a noisy channel.

\newpage
Explicit constructions of codes often have a lot of symmetry. In particular, many known constructions of codes are \emph{transitive}. The group of symmetries of a code is the subgroup $G$ of permutations $\pi:\{1,2,\dotsc, N\} \rightarrow \{1,2,\dotsc, N\}$ such that permuting the coordinates of each of the codewords using $\pi$ does not change the code. A code is transitive if for every two coordinates $i,j$, there is a permutation $\pi \in G$ with $\pi(i) = j$. A code is doubly transitive if for every $i \neq k$, $j \neq \ell$ there is a permutation $\pi \in G$ with $\pi(i) = j, \pi(k) = \ell$. Many known constructions of codes are \emph{cyclic}, and every cyclic code is transitive. Reed-Solomon codes, BCH codes and Reed-Muller codes are all transitive. In addition, Reed-Muller codes and extended primitive narrow-sense BCH codes are doubly transitive.

Using fundamental results from Fourier analysis about the influences of symmetric boolean functions \cite{kahn1988kkl, Talagrand94, Bourgain97} has led to a very successful line of work, with \cite{kudekar2016erasure} showing that Reed-Muller codes achieve capacity over the BEC and
\cite{reeves2021bitcapacity,abbe2023rmcapacityBSC}
showing that they achieve capacity over the BSC. In fact, \cite{kudekar2016erasure} show that for any $s>1$, if a linear code $C \subseteq \F_2^N$ of size
$|C|\leq 2^{(1-\epsilon)N}$ has a doubly-transitive symmetry group $G$ such that for every $S \subseteq \{1,2,\dotsc, N\}$ with $|S| = (s \log N)^{0.99}$, $|\{\pi(S): \pi \in G\}| \geq N^{s+1}$, then $C$ can tolerate $\epsilon - O(1/s)$ fraction of random erasures \footnote[1]{Their result is not stated in this form, but we believe this follows from their analysis.}. Given these results, it is natural to investigate the types of symmetry that lead to good codes. In this paper, we prove three kinds of results relevant to understanding the error resilience of general linear codes,  transitive linear codes, and doubly transitive linear codes.
\begin{enumerate}
	\item We give a clean and tight weight distribution bound for every transitive linear code. We show that for any such code $C \subseteq \F_2^N$,
$$\Pr_{c \in C}[\textnormal{wt}(c) = \alpha N] \leq 2^{-(1-h(\alpha)) \mathsf{dim}(C)},$$
where $\textnormal{wt}(c)$ denotes the Hamming weight of the binary vector $c$. This bound is proven by combining transitivity with the subadditivity of entropy. In some regimes, it improves on all previously known weight bounds for Reed-Muller codes (See Table \ref{tablerm} and Appendix \ref{aweight} for a comparison of our weight bound with previous results).

\item We give a new criterion to validate that a code can be decoded over the BSC. For any fixed integers $0\leq s\leq N$, define the Krawtchouk polynomial of degree $s$ to be the real polynomial
$$K_{s}(x):= \sum_{j=0}^s (-1)^j\binom{x}{j}\binom{N-x}{s-j},$$
where for any polynomial $p(x)$ we abused notation to write

$\binom{p(x)}{j}:=\frac{p(x)(p(x)-1)\dotsc (p(x)-j+1)}{j!}$. Let $C^\perp$ denote the dual code of $C$. In spirit, our criterion says that any linear code $C$ satisfying
$$ \E_{c \in C^\perp}[K_{\epsilon N}(\textnormal{wt}(c))^2] < (1+o(N^{-1}))  \binom{N}{\epsilon N}$$
can be uniquely decoded on the BSC with high probability. Our actual result is a little more technically involved (see Theorem \ref{fouriercriterionunique}). This criterion implies that any linear code whose dual codewords are distributed sufficiently close to the binomial distribution must be resilient to $\epsilon$-errors (see Corollary \ref{weightunique}).  Moreover, if the above expectation is bounded by $o(k\binom{N}{\epsilon N})$, then we prove that the code can be list-decoded with a list size of about $k$.

\item Finally, we combine
known estimates for the Krawtchouk polynomials with our weight bound for transitive codes, and with known weight bounds for doubly transitive codes, to obtain list-decoding results for both families of codes. In some regimes, our bounds for doubly transitive codes achieve the information-theoretic optimal trade-off between rate and list size.

\end{enumerate}

Next, we discuss our results more rigorously. Throughout this section, for any set $X$ we denote the uniform distribution over $X$ by $\mathcal{D}(X)$.

\hfill\\
\textbf{I. weight bounds for transitive codes}
\hfill\\
We bound the weight distribution of any transitive linear code over any finite field. See section \ref{weighttransitive} for the proof.
\begin{theorem}\label{probtransitive}
Consider any finite field $\F_q$, and let $C\subseteq \mathbb{F}_q^N$ be any transitive linear code. Then for any $\alpha\in (0,1)$, we have
$$ \Pr_{c\sim \mathcal{D}(C)}\Big[\textnormal{wt}(c) = \alpha N\Big] \leq q^{-(1-h_q(\alpha)) \textnormal{dim }C},$$
where $\mathcal{D}(C)$ is the uniform distribution over all codewords in $C$, $\textnormal{wt}(c)$ is the number of non-zero coordinates of $c$, and $h_q$ is the q-ary entropy $$h_q(\alpha):= (1-\alpha) \log_q \frac{1}{1-\alpha} + \alpha \log_q\frac{q-1}{\alpha}.$$
\end{theorem}
Note that $h_2(\alpha)$ denotes the binary entropy function. We note that in some regimes (see Table \ref{tablerm}), the bound above improves on all previously proven weight distribution bounds for Reed-Muller codes, even though the only feature of the code that we use is transitivity. See Appendix \ref{aweight} for a comparison of our Theorem \ref{probtransitive} with previous weight bounds.

\hfill\\
\textbf{II. a criterion for decoding on the BSC}
\hfill\\
We develop a new approach for proving decoding results over the BSC, i.e. the communication channel whose errors $z\in\F_2^N$ are sampled from the $\epsilon$-noisy distribution
$$P_\epsilon(z):=\epsilon^{\textnormal{wt}(z)}(1-\epsilon)^{N-\textnormal{wt}(z)}$$
for some $\epsilon\in(0,\frac{1}{2})$.
Our approach is based on Fourier analysis, although unlike \cite{kudekar2016erasure} and  \cite{hazla2021polyclose}, the ideas we use do not rely on bounds on influences.
We obtain the following result. 
\begin{theorem}\label{fouriercriterionunique}
Let $C\subseteq\F_2^N$ be any linear code, and denote by $C^\perp\subseteq\F_2^N$ its dual code. Then for any $\epsilon \in (0,$ $\frac{1}{2})$ satisfying $N>\frac{1}{\epsilon^4(\frac{1}{2}-\epsilon)^4} $, there exists a decoding function $d:\F_2^N\rightarrow C$
such that for all $c\in C$ we have
\begin{align*}
    \Pr_{\rho\sim P_\epsilon}[ d(c+\rho)\neq c]&\leq 2e^{-\frac{\sqrt{N}}{3\epsilon}}+N \mathop{\max}_{\substack{S\subseteq [\epsilon N\pm N^{3/4}]\cap\N \\
    1\leq|S|\leq 2}} \Big\{  \frac{1}{\binom{N}{S}}\E_{c\sim\mathcal{D}(C^\perp)}\big[ K_S(\textnormal{wt}(c))^2\big]-1 \Big\},
\end{align*}
where $\binom{N}{S}:=\sum_{j\in S}\binom{N}{j}$, and $K_S(x):=\sum_{j\in S}K_j(x)$ for $K_j$ the Krawtchouk polynomial of degree $j$, and where $[\epsilon N\pm N^{3/4}]$ denotes the interval $[\epsilon N-N^{3/4},\epsilon N+N^{3/4}]$.
\end{theorem}
See section \ref{weightdecoding} (Theorem \ref{fouriercriteriongeneral}) for the proof. We will now consider one interesting consequence of Theorem \ref{fouriercriterionunique}. Let $\epsilon\in(0,\frac{1}{2})$ be arbitrary, and define $$A_\epsilon:=\{\alpha N\in\N:h(\alpha)>1-h(\epsilon)-N^{-1/5}\}.$$
Our next corollary states that whenever the dual codewords of $C$ are distributed sufficiently close to the binomial distribution for all weights in $A_\epsilon$, the code $C$ must be resilient to $\epsilon$-errors. See Appendix \ref{abinomialweight} for the proof.
\begin{corollary}\label{weightunique}
Let $\epsilon\in(0,\frac{1}{2})$ be arbitrary, and let $C\subseteq\F_2^N$ be a linear code.
Suppose that for every $j\in A_\epsilon$ we have
\begin{align*}
    \Pr_{y\sim\mathcal{D}(C^\perp)}\big[\textnormal{wt}(y)=j\big]\leq \big( 1+o(N^{-1})  \big)\frac{\binom{N}{j}}{2^N},
\end{align*}
and suppose that
\begin{align*}
    \Pr_{y\sim\mathcal{D}(C^\perp)}\big[\textnormal{wt}(y)\notin A_\epsilon\big]\leq 2^{N^{\frac{3}{4}}}\cdot \frac{\sum_{i\notin A_\epsilon}\binom{N}{i}}{2^N}.
\end{align*}
Then $C$ is resilient to $\epsilon$-errors.
\end{corollary}
As a proof of concept, we note that a uniformly random linear code of dimension $(1-h(\epsilon))N-\sqrt{N}$ satisfies all these conditions simultaneously with high probability.

\hfill\\
\textbf{III. list decoding results}
\hfill\\
Using a generalized version of Theorem \ref{fouriercriterionunique} (namely, Theorem \ref{fouriercriteriongeneral} in section \ref{weightdecoding}), we obtain list decoding bounds for both transitive and doubly transitive codes. We start with our bound for doubly transitive codes (see Section \ref{listrm} for the proof).
\begin{theorem}\label{rmclose}
Fix any $\epsilon\in(0,\frac{1}{2})$ and any $\gamma\leq 1-\log(1+2^{-4\epsilon})$. Then any doubly  transitive linear code $C\subseteq\F_2^N$ of dimension $\textnormal{dim }C=(1-\gamma) N$ can with high probability list-decode $\epsilon$-errors using a list $T$ of size
$$|T|=2^{h(\epsilon)N-\gamma N+o(N)} .$$
\end{theorem}
Although our lists have exponential size, the list size is non-trivial in the sense that it is much smaller than the number of noise vectors (which is about $\binom{N}{\epsilon N}\approx2^{h(\epsilon)N}$) and the number of codewords in the code (which is $2^{\textnormal{dim }C}=2^{(1-\gamma)N}$). In fact, a standard calculation (see Appendix \ref{alistcapacity}) shows that any code $C\subseteq\F_2^N$ of dimension $(1-\gamma)N$ that can successfully list-decode errors of probability $\epsilon$ with list size $|T|$ must satisfy
\begin{align}\label{lwboundlist}
|T|\gtrsim 2^{(h(\epsilon)-\gamma)N}.
\end{align}
Our bound in Theorem \ref{rmclose} shows that doubly transitive codes achieve these optimal parameters, at least in some regimes. (Since the requirement $\gamma \leq 1-\log(1+2^{-4\epsilon})$ can be a bit hard to digest, we note for e.g. that $1.3\epsilon<1-\log(1+2^{-4\epsilon})$ for all $\epsilon\in(0,\frac{1}{2})$, so Theorem \ref{rmclose} implies that any doubly transitive code of rate $\geq1-1.3\epsilon$ achieves the optimal list size for decoding $\epsilon$-errors).
We now turn to our list-decoding bound for transitive codes (see section \ref{listtransitive} for the proof).
\begin{theorem}\label{thmtransitivelist}
Fix any $\epsilon\in(0,\frac{1}{2})$ and $\eta\in(0,1)$.
Then any  transitive linear code $C\subseteq\F_2^N$ of dimension $\textnormal{dim }C=\eta N$ can with high probability list-decode $\epsilon$-errors using a list $T$ of size
\begin{align*}
|T|=2^{\epsilon N\log(\frac{2}{1-\eta})+o(N)}+2^{4\epsilon N+o(N)}.
\end{align*}
\end{theorem}
As an explicit example of the types of bounds one gets from Theorem \ref{thmtransitivelist}, we have that any transitive linear code of dimension $\textnormal{dim }C=(1-\frac{4\epsilon}{e})N$
can with high probability list-decode $\epsilon$-errors using a list $T$ of size
\begin{align}\label{egtransitive}
|T|= 2^{(h(\epsilon) -\epsilon +\frac{\epsilon^2}{\ln2} )N+o(N)}+2^{4\epsilon N+o(N)}.
\end{align}
See Appendix \ref{calculationsthm} for the calculations. For comparison, recall that the lower bound (\ref{lwboundlist}) states that any code $C$ of dimension $(1-\frac{4\epsilon}{e})N$ requires a list size of at least about $ 2^{(h(\epsilon)-\frac{4\epsilon}{e})N}$.
\subsection{Techniques}
Our weight distribution bound for transitive linear codes (Theorem \ref{probtransitive}) is proven by showing that the entropy of a uniformly random codeword of weight $\alpha N$ is small. To do this, we analyze the entropy of the coordinates corresponding to linearly independent columns of the generator matrix. Transitivity implies that every coordinate in the code has the same entropy, and subadditivity of entropy can then be used to bound the entropy of the entire distribution.

To obtain our decoding criterion, we make use of a connection between the probability of a decoding error and the $\ell_2$ norm of the coset distribution of the code. To explain the intuition, let us start by assuming that exactly $\epsilon N$ of the coordinates in the codeword are flipped, although our results actually hold over the BSC as well. Let $z$ be the vector in  $\F_2^N$ that represents the errors introduced by the channel, and let $H$ be the parity check matrix of the code. Then by standard arguments, if $z$ can be recovered from $Hz^\intercal$, the codeword can be decoded. In the case where $z$ is uniformly distributed on vectors of weight $\epsilon N$, this amounts to showing that
with high probability, the coset of $z$ does not contain any other string of weight $\epsilon N$ (in other words, there is no $w\in\F_2^N$, $w\neq z$ of weight $\textnormal{wt}(w)=\epsilon N$ such that $Hz^\intercal=Hw^\intercal$). This can be understood by computing the norm $$\|f\|_2^2 := \sum_{y} f(y)^2 = \sum_{y} \Pr[Hz^\intercal = y^\intercal]^2,$$where $f(y) = \Pr[Hz^\intercal = y^\intercal]$. This norm computes the probability that two independent, uniformly random strings $z,z'$ of weight $\epsilon N$ collide under the mapping $z\mapsto Hz^\intercal.$ Thus $\|f\|_2^2$ is always at least $\binom{N}{\epsilon N}^{-1}$, because with probability $\binom{N}{\epsilon N}^{-1}$ we have $z=z'$. If $\|f\|_2^2$ is close to $\binom{N}{\epsilon N}^{-1}$, then the code can be decoded with high probability. If $\|f\|_2^2$ is larger than $\binom{N}{\epsilon N}^{-1}$, then we show that the code can be list-decoded with high probability, where the size of the list is proportional to $\binom{N}{\epsilon N} \|f\|_2^2$.

Thus, to understand decoding, we need to understand $\|f\|_2^2$. Using Fourier analysis, we express this quantity as
\begin{align}\label{introlink}
 \|f\|_2^2 = \frac{1}{\binom{N}{\epsilon N}^2}\sum_{j=0}^N \Pr[\textnormal{wt}(c^\perp) = j] \cdot K_{\epsilon N}(j)^2,
 \end{align}
where $c^\perp$ is a uniformly random codeword in the dual code and {\small $K_{\epsilon N}$ is the Krawtchouk} polynomial of degree $\epsilon N$.
We note that such relations for the coset weight distribution have been used to understand the discrepancy of subsets of the sphere, as well as subsets of other homogeneous spaces. In particular, (\ref{introlink}) was proven in a slightly different form in \cite{2021bargfourierlink} (see Theorem 2.1 and Lemma 4.1), whereas over $\R^N$ results of this type had previously been derived in \cite{bilyk2018stolarskyprinciple,skriganov2019stolarskyhomogeneous}.

Using estimates for the magnitude of Krawtchouk polynomials and bounds for the weight distribution of the dual code $C^\perp$, one can thus bound the norm $\|f\|_2^2$ in the set-up where the error string $z$ is a random vector of weight exactly $\epsilon N$. Using essentially the same techniques, one can also bound the norm $\|f\|_2^2$ when the error string $z$ is a random vector of weight $\approx\epsilon N$, i.e. $z$ is taken uniformly at random from the set $S=\{x\in\F_2^N:\textnormal{wt}(x)=\epsilon N\pm N^{3/4}\}$.

Our next step is then to show that the $\ell_2$ norm corresponding to the $\epsilon$-noisy distribution is very similar to the $\ell_2$ norm corresponding to the uniform distribution over $S$. Intuitively, this is because $S$ only contains a very small range of weights, so the $\epsilon$-noisy distribution and the uniform distribution must behave very similarly over strings of weight in $S$.
It then follows that their corresponding $\ell_2$ norms must be similar as well.

Our decoding criteria (Theorem \ref{fouriercriterionunique}, Corollary \ref{weightunique}) are thus obtained by bounding the norm $\|f\|_2^2$ using estimates for Krawtchouk polynomials and for the weight distribution of the dual code $C^\perp$. Our list-decoding results (Theorems \ref{rmclose} and \ref{thmtransitivelist}) then follow from our weight bound for transitive codes (Theorem \ref{probtransitive}) and from a weight bound of Samorodnitsky for doubly transitive codes (Theorem \ref{previousboundsmall}).

\subsection{Related work}
It has been shown that LDPC codes achieve capacity over  Binary Memoryless Symmetric Channels (BMS) \cite{luby1997ldpc2,kudekar2013ldpc,gallager1962ldpc}, which
includes both the BSC and the BEC. These constructions are not deterministic, and it is only with the advent of polar codes \cite{arikan2009polar} that we obtained  capacity-achieving codes with both a deterministic constructions and efficient encoding and decoding algorithms.

Polar codes are closely related to Reed-Muller codes, in the sense that they also consist of subspaces that correspond to polynomials over $\F_2$ \cite{arikan2009polar}. For this reason, when Arikan showed that polar codes achieve capacity over the BSC, Reed-Muller codes received renewed attention from the coding theory community. A long and fruitful line of work \cite{abbe2015RMlowrate,kudekar2016erasure,abbe2019rmpolarize,hazla2021polyclose,abbe2020almostRM,reeves2021bitcapacity,samorodnitsky2022undetected} has recently culminated in Abbe and Sandon showing that Reed-Muller codes achieve capacity over all BMS channels \cite{abbe2023rmcapacityBSC}.

One of the key properties of Reed-Muller codes, which is strongly leveraged in all the papers above, is that they are doubly transitive. In fact, Kudekar, Kumar, Mondelli, Pfister, Sasoglu and Urbanke showed that any doubly transitive linear code achieves bit-decoding capacity over the BEC \cite{kudekar2016erasure}, i.e. that one can with high probability recover any single bit of the original codeword (but not with high enough probability that one could take a union bound). An important open question is thus whether general doubly transitive codes achieve capacity over all BMS channels under block-MAP decoding, or whether one really needs the additional symmetry that Reed-Muller codes possess. Some of the key techniques used in \cite{reeves2021bitcapacity} and \cite{abbe2023rmcapacityBSC} are very much tailored to Reed-Muller codes, or at least to codes consisting of evaluations of polynomials over $\F_2^N$; in order to prove the same results for arbitrary doubly transitive codes, it may be necessary to develop a more general framework.

\hfill\\
\textbf{Weight bounds for doubly transitive codes}
\hfill\\
As far as we know, there were no previously known weight bound for general transitive linear codes. There are however two known weight bounds for doubly transitive codes (which we'll give here), as well as many known weight bounds for Reed-Muller codes (which we'll discuss in the next section). We compare all these results in Appendix \ref{aweight}. We state below the weight bounds of Samorodnitsky, which to the best of our knowledge are the only previously known weight bounds for doubly transitive codes.

\begin{theorem}[Proposition 1.4 in \cite{samorodnitsky2020weightimproved}]\label{previousboundsmall}
Let $C\subseteq\F_2^N$ be a doubly transitive linear code of rate $\eta:=\frac{\textnormal{dim }C}{N}$. For any $j\in\{1,2,\dotsc ,N\}$, define $j^{*}:=\min\{j,N-j\}.$  Then for any $j\in\{1,2,\dotsc ,N\}$,
$$\left|\Big\{c\in C:\textnormal{wt}(c)= j\Big\}\right|\leq 2^{o(N)}\cdot
\left(\frac{1}{2^{1-\eta}-1}\right)^{j^*}.
$$
Moreover, if $j^*\geq (1-2^{\eta-1})N$,
$$\left|\Big\{c\in C:\textnormal{wt}(c)= j\Big\}\right|\leq 2^{o(N)}\cdot
\frac{\binom{N}{j^*}|C|}{2^N}.
$$
\end{theorem}

\hfill\\
\textbf{Weight bounds for reed-muller codes}
\hfill\\
As we mentioned earlier, one specific family of doubly transitive codes that has received a lot of attention is the family of Reed-Muller codes. Several past works have proven bounds on their weight distribution. We give here a brief history of these results, although for space reasons (there are over 10 different weight bounds), we will not state them here. We delve deeper into some prior results in Appendix \ref{aweight}, where we compare them to our weight bound of Theorem \ref{probtransitive}. We also refer the reader to \cite{abbe2021survey,abbe2023survey2} for a discussion on the subject, as well as a thorough exposition to Reed-Muller codes.

The earliest work we are aware of is that of Sloane and Berlekamp, who characterized all codewords in Reed-Muller codes of degree $2$ \cite{sloane1970degree2}. For arbitrary degree, Kasami and Tokura then characterized all codewords of weight smaller than twice the minimum distance \cite{kasami1970distance2}, before Kasami, Tokura and Azumi improved this characterization to include all codewords of weight up to 2.5 times the minimum distance \cite{kasami1976distance2half}.

\begin{table}
\begin{adjustbox}{max width=1\textwidth,center}%[t]
%\begin{center}
	\begin{threeparttable}
	\small
	\renewcommand{\arraystretch}{3} %This controls vertical stretch
 \caption{Best known upper bounds on the number of codewords of weight $w$ in $\mathsf{RM}(n,d)$}\label{tablerm}
\begin{tabular}{cccc}
\toprule
& $\bm{o(n)\leq d\lesssim 0.38n}$ & $\bm{0.38n\lesssim d\leq \frac{n}{2}-\Tilde{\Omega}(\sqrt{n})}$ &    $\bm{ d= \frac{n}{2}\pm \Tilde{O}(\sqrt{n})}$ \\ \midrule
$\bm{o(N)\leq w< \tau N}$* & $2^{O\big(\binom{n}{\leq d}(\frac{d}{n})^{\lceil\log\frac{N}{\textnormal{w}}\rceil}\log\frac{N}{\textnormal{w}}\big)}$ \cite{sberlo2020weightbound}& $\leftarrow$ as previous  & $
\left(\frac{1}{2^{1-\binom{n}{\leq d}/N}-1}\right)^{w+o(N)}$\cite{samorodnitsky2020weightimproved}** \\
$ \bm{\tau N\leq w< \frac{N}{4}}$ &  $ 2^{h(\frac{w}{N})\binom{n}{\leq d}}\textnormal{ (Theorem \ref{probtransitive})}$& $\leftarrow$ as previous & as above $\uparrow$  \\
$ \bm{\frac{N}{4}\leq w\leq \frac{N}{2}-o(N)}$ & $2^{(1-2^{-O(\log\frac{N}{w})})\binom{n}{\leq d}}$\cite{sberlo2020weightbound}&  $ 2^{h(\frac{w}{N})\binom{n}{\leq d}}\textnormal{ (Theorem \ref{probtransitive})}$ & as above $\uparrow$\\
 \hdashline
$ \bm{(1-2^{\binom{n}{\leq d}/N-1})N\leq w\leq \frac{N}{2}}$ &  &
 & $\binom{N}{w}\cdot2^{\binom{n}{\leq d}-N+o(N)}$ \cite{samorodnitsky2020weightimproved}\\
\bottomrule
%&*unless $N\geq$, in which case see row below dashed line
\end{tabular}
		\begin{tablenotes}
                \item *$\tau$ is a threshold that depends on $\frac{d}{n}$. See (\ref{taueqn}) and the surrounding discussion.
			\item **Unless $w\geq(1-2^{\frac{\binom{n}{\leq d}}{N}-1})N$, in which case see row below dashed line.
		\end{tablenotes}
%\caption{Best known upper bounds on the number of codewords of weight $w$ in $\mathsf{RM}(n,d)$}\label{tablerm}
\end{threeparttable}
%\end{table}
%\end{center}
\end{adjustbox}
\end{table}

A few decades later, Kaufman, Lovett and Porat gave asymptotically tight bounds on the weight distribution of Reed-Muller codes of constant degree \cite{kaufman2012constantdegree}. Abbe, Shpilka and Wigderson then built on these techniques to obtain bounds for all degrees smaller than $\frac{n}{4}$ \cite{abbe2015RMlowrate}, before Sberlo and Shpilka again improved the approach to obtain bounds for all degrees \cite{sberlo2020weightbound}. Most recently, Samorodnitsky used completely different ideas to obtain weight bounds for codes of constant rate \cite{samorodnitsky2020weightboundhalf,samorodnitsky2020weightimproved} (see previous section).

The bounds mentioned above are strong when $j/N \ll 1/2$. For $j/N$ close to $1/2$, the first results we are aware of are due to Ben-Eliezer, Hod and Lovett \cite{ben-eliezer2012weighthalf1}. Their bounds were extended to Reed-Muller codes over prime fields by Beame, Oveis Gharan and Yang \cite{beame2020weightodd}. Sberlo and Shpilka then obtained the first results to hold for all degrees in \cite{sberlo2020weightbound}, while Samorodnitsky again obtained bounds for codes of constant rate in \cite{samorodnitsky2020weightimproved}.

\medskip
We summarize in Table \ref{tablerm} the best known upper bounds on the weight distribution of Reed-Muller codes.
We note that in some regimes, our Theorem \ref{probtransitive} improves on all the aforementioned weight bounds. See Appendix \ref{aweight} for some details; see also \cite{abbe2023survey2}, section 4.

\medskip
In Table \ref{tablerm}, $\tau$ is a threshold that depends on $\frac{d}{n}$. We show for e.g. that
\begin{align}\label{taueqn}
    \tau\leq\frac{1}{2}\cdot 2^{-\frac{\log17}{\log\frac{n}{2d}}}
\end{align} 
 in Appendix \ref{aweight}, which is below the trivial $\frac{1}{4}$ for any $\frac{d}{n}>\frac{1}{34}$. We note that when $\frac{d}{n}$ is small enough (smaller than some constant), then $\tau=\frac{1}{4}$.

\hfill\\
\textbf{List decoding}
\hfill\\
List decoding was proposed by Elias in 1957 as an alternative to unique decoding \cite{elias1957firstlist}. In the list decoding framework, the receiver of a corrupted codeword is asked to output a list of potential codewords, with the guarantee that with high probability one of these codewords is the original one. This of course allows for a greater fraction of errors to be tolerated.

\medskip
The list decoding community has largely focused on proving results for the adversarial noise model, and many codes are now known to achieve list-decoding capacity. For example uniformly random codes achieve capacity, as do uniformly random linear codes \cite{guruswami2002listlinear1,Li2018listlinear2,Guruswami2011listlinear3}. Folded Reed-Solomon codes were the first explicit codes to provably achieve list-decoding capacity \cite{Guruswami2008explicitlist}, followed by several others a few years later \cite{Guruswami2012explicitlist2,Kopparty2015explicitlist3,Hemenway2017explicitlist4,Mosheiff2020explicitlist5,brakensiek2022listdecoding}.
For the rest of this paper however, we will exclusively work in the model where the errors are stochastic. %We compare in Table \ref{tablelist} our Theorems \ref{} and \ref{}
In this model, as far as we know, there was no known list-decoding bound for transitive codes prior to our Theorem \ref{thmtransitivelist}. For doubly transitive codes,
the strongest previously known list decoding bound was, to the best of our knowledge, that any doubly transitive code $C\subseteq\F_2^N$ of dimension $\textnormal{dim }C=(1-\gamma) N$ can list-decode $\epsilon$-errors with a list $T$ of size
\begin{align}\label{previouslistresult}
    |T|=2^{\epsilon N\log\frac{4\epsilon(1-\epsilon)}{(2^\gamma-1)^2}+o(N)}.
\end{align}
%and succeed with high probability in decoding $\epsilon$-errors.
This result, although not
explicitly stated in \cite{samorodnitsky2020weightimproved}, can be obtained from his weight bound of Theorem \ref{previousboundsmall} by bounding the
expected number of codewords that end up closer to the received string than the original codeword, and then applying Markov’s inequality.
We summarize in Table \ref{tablelist} our list-decoding results and compare them to previous work.
\begin{table}
\begin{adjustbox}{max width=1\textwidth,center}
	\begin{threeparttable}[!htbp]
	\small
	\renewcommand{\arraystretch}{3} %This controls vertical stretch
\caption{\small Upper bounds on the list size needed for a code of rate $1-\gamma $ to recover from $\epsilon$-errors}\label{tablelist}
\begin{tabular}{cccc}
\toprule
 & \textbf{Previous work} &    \textbf{Our results} & \makecell{\textbf{Information-theoretic}\\ \textbf{lower bound}}\\ \midrule
\makecell{\textbf{Transitive codes}\\(any $\gamma$)}&- & $2^{\epsilon N\log(\frac{2}{\gamma})+o(N)}+2^{4\epsilon N+o(N)}$ & $2^{h(\epsilon)N-\gamma N-o(N)}$\\
\makecell{\textbf{Doubly transitive codes}\\(any $\gamma\leq 1- \log(1+2^{-4\epsilon})$)} & $2^{\epsilon N\log\frac{4\epsilon(1-\epsilon)}{(2^\gamma-1)^2}+o(N)}$ & $2^{h(\epsilon)N-\gamma N+o(N)}$  & $2^{h(\epsilon)N-\gamma N-o(N)}$   \\
\bottomrule
\end{tabular}
		
\end{threeparttable}
\end{adjustbox}
\end{table}
We note that the previously known bound for doubly transitive codes stays strictly above the optimal size of $2^{h(\epsilon)N-\gamma N\pm o(N)}$ (see Appendix \ref{acomparelist}).

\medskip
Following the publication of the present paper on arxiv, Hazla showed in \cite{hazla2022exponentiallist} that any code $C\subseteq\F_2^N$ of dimension $\textnormal{dim }C=(1-\gamma)N\geq (1-4\epsilon(1-\epsilon))N$ that achieves capacity over the BEC can list-decode $\epsilon$-errors with a list $T$ of size
$$|T|=2^{\gamma N-h(\epsilon)N+o(N)}.$$

\hfill\\
\textbf{Krawtchouk polynomials}
\hfill\\
Fix any non-negative integers $N$ and $s\leq N$. The Krawtchouk polynomial of degree $s$ is the real polynomial
$$K_s(x):=\sum_{j=0}^s(-1)^j\binom{x}{j}\binom{N-x}{s-j},$$
where for any polynomial $p(x)$ we defined $\binom{p(x)}{j}:=\frac{p(x)(p(x)-1)\dotsc (p(x)-j+1)}{j!}$. For any subset $S\subseteq \{0,1,\dotsc,N\}$, we will be interested in the polynomial $K_S(x):=\sum_{s\in S}K_s(x)$. For $v\in\F_2^N$, we will sometimes abuse notation and use $K_S(v)$ to mean $K_S(\textnormal{wt}(v))$, where $\textnormal{wt}(v)$ denotes the Hamming weight of $v$. The following proposition follows from standard results (see for instance \cite{1999surveykrawtchouk}, or Theorem 16 in \cite{1977bookkrawtchouk}).
\begin{prop}\label{IFourier}
For any $N$ and any $S\subseteq\{0,1,\dotsc,N\}$, we have
   $$\frac{2^{-N}}{\sum_{s\in S}\binom{N}{s}}\sum_{j = 0}^N  \binom{N}{j}  K_S(j)^2=1.$$
\end{prop}
Good estimates for Krawtchouk polynomials of any degree were obtained in \cite{kallai1995krawtchouk1,ismail1998krawtchouk2,polyanskiy2019krawtchouk3} (see for e.g. \cite{polyanskiy2019krawtchouk3}, Lemma 2.1). These estimates are asymptotically tight in the exponent. Note that $|K_s(x)|=| K_s(N-x)|= |K_{N-s}(x)|$ by symmetry (see for e.g. equations (2.8) and (2.9) in \cite{polyanskiy2019krawtchouk3}), so it suffices to understand the case $x,s\leq\frac{N}{2}$.
\begin{theorem}[\cite{kallai1995krawtchouk1,ismail1998krawtchouk2,polyanskiy2019krawtchouk3}]\label{thmkrawtchouktight}
Let $\epsilon,\delta\in(0,\frac{1}{2})$ be arbitrary. If $\delta\geq\frac{1}{2}- \sqrt{\epsilon(1-\epsilon)}$, then
$$|K_{\epsilon N}(\delta N)|\leq 2^{(1+h(\epsilon)-h(\delta))\frac{N}{2}}.$$
If $\delta<\frac{1}{2}- \sqrt{\epsilon(1-\epsilon)}$, define $\omega:=\frac{1-2\delta-\textnormal{sgn}(1-2\delta)\sqrt{(1-2\delta)^2-4\epsilon(1-\epsilon)}}{2(1-2\delta)}.$ Then
$$|K_{\epsilon N}(\delta N)|\leq \frac{(1-\omega)^{\delta N}(1+\omega)^{(1-\delta)N}}{\omega^{\epsilon N}}.$$
\end{theorem}
As the second expression can be somewhat cumbersome to use, \cite{polyanskiy2019krawtchouk3} also gives the following weaker bound.
\begin{theorem}[Lemma 2.2 and equation 2.10 in \cite{polyanskiy2019krawtchouk3}]\label{thmkrawtchoukbound}
For any $\epsilon \in(0,\frac{1}{2})$ and any $\delta<\frac{1}{2}-\sqrt{\epsilon(1-\epsilon)}$, we have
$$|K_{\epsilon N}(\delta N)|\leq 2^{h(\epsilon)N+\epsilon N\log(1-2\delta)}.$$
\end{theorem}
We will need the above estimate when using our Theorem \ref{fouriercriterionunique} to obtain list-decoding results for transitive and doubly transitive codes.

\hfill\\
\textbf{Relations between a code and its dual}
\hfill\\
Several connections have been established between the properties of a code $C\subseteq\F_2^N$ and those of its dual $C^\perp$. MacWilliams proved in \cite{Macwilliams1963identity} the MacWilliams identities, relating the weight distributions of $C$ and $C^\perp$ by %the identity
\begin{align*}
    \sum_{c\in C}(1+z)^{N-\textnormal{wt}(c)}(1-z)^{\textnormal{wt}(c)}=|C|\sum_{c\in C^\perp} z^{\textnormal{wt}(c)},
\end{align*}
where $z$ is an indeterminate.
Krasikov and Litsyn then bounded the weight distribution of any linear code with large dual distance \cite{krasikov1995weight,krasikov1998dualdistancetoweight}, while Ashikhmin, Honkala, Laihonen and Litsyn derived bounds for the covering radius of any such code \cite{ashikhmin1999coveringradius}. To the best of our knowledge however, the present paper is the first work to relate the decoding performance of a code $C\subseteq\F_2^N$ to the weight distribution of its dual% $C^\perp$
. As far as we know, there is no known way to apply the results mentioned above to obtain a unique-decoding criterion like our Theorem \ref{fouriercriterionunique} or our Corollary \ref{weightunique}.

\section{Outline of the paper}\label{outline}
The main question we will be looking into is whether or not a family of list-decoding codes $\{C_N\}$, with $C_N\subseteq\mathbb{F}_2^N$, is asymptotically resilient to independent errors of probability $\epsilon$. Formally, we are given a list size $k=k(N)$ and want to know if there exists a family of decoding functions $\{d_N\}$, with $d_N:\mathbb{F}_2^N\rightarrow \left(\mathbb{F}_2^N\right)^{\otimes k}$, such that for every sequence of codewords $\{c_N\}$ we have
\begin{align*}
    \lim_{N\rightarrow \infty}\Pr_{\rho_N\sim P_\epsilon} \big[c_N \notin d_N(c_N+\rho_N)\big] = 0.
\end{align*}
We note that the unique decoding problem can be seen as setting $k=1$ in the above set-up.
Our general approach will be based on trying to identify the error string $\rho\in\F_2^N$ from its image $H\rho^\intercal$. In particular, we will be interested in the max-likelihood decoder
\begin{align}\label{maxdecoderlist}
D_k(x)&:=\mathop{\textnormal{argmax}}_{\substack{\{z_1,z_2\dotsc,z_k\}\subseteq \mathbb{F}_2^N\\{Hz_i}^\intercal=x^\intercal \textnormal{ for all }i}}\{P_\epsilon(z_1)+P_\epsilon(z_2)+\dotsc+P_\epsilon(z_k)\} \nonumber \\
&=\mathop{\textnormal{argmin}}_{\substack{\{z_1,z_2\dotsc,z_k\}\subseteq \mathbb{F}_2^N\\Hz_i^\intercal=x^\intercal \textnormal{ for all }i}}\{\textnormal{wt}(z_1)+\textnormal{wt}(z_2)+\dotsc+\textnormal{wt}(z_k)\},
\end{align}
where ties are broken according to the lexicographic order. The following standard lemma (see for e.g. page 17, Theorem 5 in \cite{1977bookkrawtchouk}) states that if the max-likelihood decoder is able to identify the error string $\rho$, then it is possible to recover the original codeword.
\begin{lemma}\label{eqvltdecoder}
Let $H$ be the $t\times N$ parity-check matrix of the linear code $C$, and let $D:\mathbb{F}_2^t\rightarrow \left(\mathbb{F}_2^N\right)^{\otimes{k}}$
be arbitrary. Then there exists a decoder $$d:\mathbb{F}_2^N\rightarrow C^{\otimes{k}}$$ such that for every $c\in C$ we have
\begin{align*}
\Pr_{\rho\sim P_\epsilon}[c\notin d(c+\rho)]\leq \Pr_{\rho\sim P_\epsilon}[\rho\notin D(H\rho^\intercal)].
\end{align*}
\end{lemma}
From this point onward, our goal will thus be to prove that the max-likelihood decoder in (\ref{maxdecoderlist}) succeeds in recovering $\rho$ with high probability.
In section \ref{collisionsdecoding}, we relate the decoding error probability of the max-likelihood decoder $D_k$ to the collision probability $$\sum_{x\in\mathbb{F}_2^t}\Pr[Hz^\intercal=x^\intercal]^2.$$
In section \ref{weightdecoding}, we build on this result to obtain a bound on the performance of $D_k$ in terms of the weight distribution of the dual code.
We then present new bounds on the weight distribution of transitive codes in section \ref{weighttransitive}. These bounds are interesting in their own right, and we show that they are essentially tight.
In section \ref{listtransitive}, we combine these bounds with our results from section \ref{weightdecoding} to obtain list-decoding results for transitive linear codes. We then repeat this argument with Samorodnitsky's Theorem \ref{previousboundsmall} in section \ref{listrm} to obtain stronger list-decoding bounds for doubly transitive codes. 

See Figure \ref{fig} for a description of the connections between our various propositions and theorems.
\begin{figure}[tbp]
\caption{Organization of our paper and connections between our results.}
\label{fig}
\centerline{\includegraphics[scale=.31]{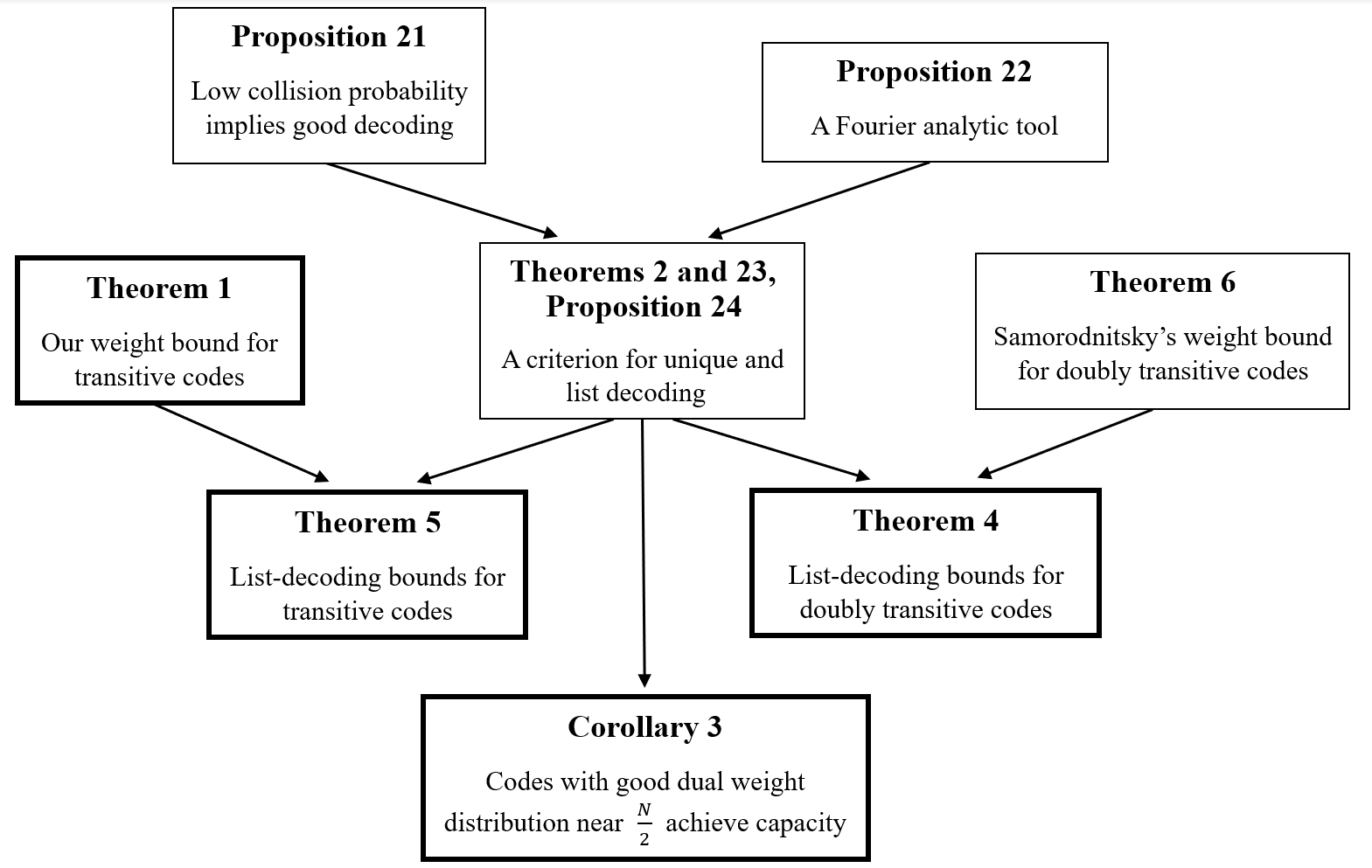}}
\end{figure}

\section{Notation, conventions and preliminaries}\label{prelim}

For the sake of conciseness, we will use the notation
$$[a\pm l]:=[a-l,a+l]$$
for intervals, the notation
$$\binom{n}{\leq d} :=\binom{n}{0} + \binom{n}{1} + \dotsb + \binom{n}{d}$$
for binomial coefficients, and for $S\subseteq\{0,1,\dotsc,N\}$ the notation
$$\binom{N}{S}:=\sum_{s\in S}\binom{N}{s}.$$
We denote the set of all non-negative integers by
$$\N:=\{0,1,2,\dotsc\}.$$
Let $N = 2^n$.
We will be working with the vector spaces $\F_2^n$ and $\F_2^N$. For convenience, we associate $\F_2^n$ with the set $[N] := \{1,2,\dotsc, N\}$, by ordering the elements of $\F_2^n$ lexicographically.
For $x \in \F_2^N$, we write $$\textnormal{wt}(x) := |\{j \in [N]: x_j=1\}|$$ to denote the Hamming weight of $x$.

\subsection{Coding theory definitions and terminology}
An $N$-bit code is a subset $C\subseteq\mathbb{F}_2^N$, and we call any element $c\in C$ a \emph{codeword} of $C$. Throughout the paper, we will use $N$ to denote the length of the code, i.e. the number of bits in any given codeword.

We will be interested in the performance of various codes over the so-called Binary Symmetric Channel (BSC for short).
When a codeword $c\in C$ is sent through the Binary Symmetric Channel, each one of its bits is flipped independently at random with probability $\epsilon$, for some $\epsilon\in(0,\frac{1}{2}).$ Throughout the paper, we will use $\epsilon$ to denote this error probability, and we will use $\rho$ to denote the vector $(\rho_1,\rho_2,\dotsc,\rho_N)$ whose $i^\textnormal{th}$ coordinate is $1$ with probability $\epsilon$ and $0$ with probability $1-\epsilon$, for all $i\in\{1,2,\dotsc N\}$. We will call the original codeword $c\in C$ the \emph{transmitted codeword}, we will call the noisy vector $\rho$ the \emph{error string}, and we will call $c+\rho$ the \emph{received message}.

We say that the code $C$ is \emph{resilient to $\epsilon$-errors} if there exists a decoding function $d:\F_2^N\rightarrow C$ such that for every $c\in C$, with high probability over the choice of an $\epsilon$-noisy error string $\rho$ we have
$$d(c+\rho)=c.$$
We will also be interested in the performance of a code with respect to list decoding. In this set-up, the decoder is now a function $d:\mathbb{F}_2^N\rightarrow C^{\otimes k}$. We say that a code $C$ can list-decode $\epsilon$-errors with a list size of $k$ if with high probability (again, over the choice of an $\epsilon$-noisy error string $\rho$), we have
$$c\in d(c+\rho).$$
Throughout the paper, we will denote by $k$ the size of the list. We note that the unique decoding problem can be seen as setting $k=1$ in the list decoding set-up.

\subsection{Linear codes}
An $N$-bit code is a subset $C\subseteq\mathbb{F}_2^N$. Whenever $C$ is a subspace of $\mathbb{F}_2^N$, we say that $C$ is a \emph{linear} code. Any linear code $C\subseteq\mathbb{F}_2^N$ can be represented by its generator matrix, which is a $\dim\textnormal{ C}\times N$ matrix $G$ whose rows form a basis of $C$. The matrix $G$ generates all codewords of $C$ in the sense that
$$C=\{vG:v\in\mathbb{F}_2^{\textnormal{dim }C}\}.$$
Another useful way to describe a linear code $C\subseteq \mathbb{F}_2^N$ is via its parity-check matrix, which is an $(N-\dim\textnormal{ C})\times N$ matrix $H$ whose rows span the orthogonal complement of $C$. The linear code $C$ can then be expressed as
$$C=\{c\in\mathbb{F}_2^N:Hc^\intercal=0\}.$$
One property that will play an important role in our analysis is transitivity, which we define below.
\begin{definition}\label{deftransitive}
A code $C\subseteq \mathbb{F}_2^N$ is transitive if for every $i,j\in[N]$ there exists a permutation $\pi:[N]\rightarrow [N]$ such that \begin{enumerate}[label=(\roman*)]
    \item $\pi(i)=j$
    \item For every element $v=(v_1,v_2,\dotsc,v_N)\in C$ we have $(v_{\pi(1)},v_{\pi(2)},\dotsc,v_{\pi(N)})\in C$.
\end{enumerate}
\end{definition}
Many well-known and widely used codes are transitive, for e.g. Reed-Muller codes, Reed-Solomon codes, general BCH codes, and all cyclic codes. In addition, Reed-Muller codes and extended primitive narrow-sense BCH codes are doubly transitive.
\begin{definition}\label{def2transitive}
A code $C\subseteq \mathbb{F}_2^N$ is doubly transitive if for every $i,j,k,\ell\in[N]$ with $i\neq k$ and $j\neq \ell$, there exists a permutation $\pi:[N]\rightarrow [N]$ such that \begin{enumerate}[label=(\roman*)]
    \item $\pi(i)=j$ and $\pi(k)=\ell$
    \item For every element $v=(v_1,v_2,\dotsc,v_N)\in C$ we have $(v_{\pi(1)},v_{\pi(2)},\dotsc,v_{\pi(N)})\in C$.
\end{enumerate}
\end{definition}
For a review on doubly transitive codes, see \cite{ivanov2020review2transitive} We note that the dual code of a transitive code is transitive, and that the dual code of a doubly transitive code is doubly transitive (see Appendix \ref{adualtransitive} for the proof).
\begin{claim}\label{dualtransitive}
The dual code $C^\perp$ of a transitive code $C\subseteq\F_2^N$ is transitive.
\end{claim}
\begin{claim}\label{dual2transitive}
The dual code $C^\perp$ of a doubly transitive code $C\subseteq\F_2^N$ is doubly transitive.
\end{claim}

\subsection{Reed-muller codes}
We will denote by $\mathsf{RM}(n,d)$ the  Reed-Muller code with $n$ variables and degree $d$. The codewords of the Reed-Muller code $\mathsf{RM}(n,d)$ are the evaluation vectors (over all points in $\F_2^n$) of all multivariate polynomials of degree $\leq d$ in $n$ variables. The dimension of the code is known to be $\binom{n}{\leq d}$. (See for e.g. page 5 of \cite{abbe2021survey}).
\begin{fact}\label{dimrm}
The dimension of the Reed-Muller code $\mathsf{RM}(n,d)$ is
\begin{align*}
\textnormal{dim}\Big( \mathsf{RM}(n,d) \Big)=\binom{n}{\leq d}.
\end{align*}
\end{fact}
Throughout this section, we let $M$ be the generator matrix of $\mathsf{RM}(n,d)$; this is an $ \binom{n}{\leq d}\times N$ matrix whose rows are indexed by subsets of $[N]$ of size at most $d$, and whose columns are indexed by elements of $\F_2^n$. For $S \subseteq [n], |S| \leq d$ and $x \in \F_2^n$, the entry of $M$ whose row is indexed by $S$ and whose column is indexed by $x$ is $$M_{S,x}:=\prod_{j \in S} x_j.$$ If $S$ is empty, this entry is set to $1$.
The parity-check matrix of the Reed-Muller code is known to be the same as the generator matrix of a different Reed-Muller code. Namely, let $H$ be the $ \binom{n}{\leq n-d-1}\times N$ generator matrix for the code $\mathsf{RM}(n,n-d-1)$. Then $H$ has full rank, and $M H^\intercal  = 0$. So, the rows of $H$ are a basis for the orthogonal complement of the span of the rows of $M$.
Reed-Muller codes also have well-known algebraic features, notably transitivity (see for e.g. Lemma 23 in \cite{kudekar2016erasure}).
\begin{fact}\label{transitive}
For all non-negative integers $n$ and $d\leq n$, the Reed-Muller code $\mathsf{RM}(n,d)$ is transitive.
\end{fact}

 \subsection{Entropy}
The binary entropy function $h:[0,1] \rightarrow [0,1]$ is defined to be
$$h(\epsilon) := \epsilon \cdot \log\frac{1}{\epsilon} + (1-\epsilon) \cdot \log\frac{1}{1-\epsilon}.$$
One useful property of the binary entropy function is that it is subadditive.
\begin{lemma}\label{entropysubadditive}
    For any $x\in[0,1]$ and any $y\in[0,1-x]$, we have
    \begin{align*}
        h(x+y)\leq h(x)+h(y).
    \end{align*}
\end{lemma}
This is because the binary entropy function is concave, and any concave, positive function is subadditive (see for e.g. \cite{hardy1934inequalities}, page 83, statement 103).
The entropy function can be used to approximate binomial coefficients.
\begin{lemma} \label{stirling} For any integer $d\in\{1,2,\dotsc,\frac{n}{2}\} $, we have
$$\frac{1 }{\sqrt{2n}} \cdot 2^{h(d/n) \cdot n} \leq \binom{n}{d}\leq \binom{n}{\leq d}  \leq 2^{h(d/n) \cdot n}.$$
\end{lemma}
See for e.g. page 309, Lemma 7 in \cite{1977bookkrawtchouk} for the proof of the leftmost inequality, and Theorem 3.1 in \cite{binomialbound} for the proof of the rightmost inequality.

The following lemma, which is essentially a 2-way version of Pinsker's inequality, gives a useful way to bound the entropy function near $1/2$.
\begin{lemma}\label{pinsker}
For any $\mu\in(0,1)$, we have $$\frac{\mu^2}{2\ln2}\leq1-h\left(\frac{1-\mu}{2}\right)\leq\mu^2.$$
\end{lemma}
See Appendix \ref{apinsker} for the proof.
\subsection{Probability distributions}
There are two types of probability distributions that we will use frequently. The first one is the $\epsilon$-Bernoulli distribution over $\mathbb{F}_2^N$, which we will denote by
$$P_\epsilon (z):=\epsilon^{\textnormal{wt}(z)}(1-\epsilon)^{N-\textnormal{wt}(z)}.$$
The second one is the uniformly random distribution over some set $T$, which we will denote by
$$\mathcal{D}(T)(z):=\begin{cases}
\frac{1}{|T|} & \text{if $z\in T$,}\\
0 & \text{otherwise.}
\end{cases}.$$
There are two particular cases for the uniform distribution that will occur often enough that we attribute them their own notation. The first one is the uniform distribution over $\mathbb{F}_2^t$, which we will denote by
$$\mu_t:=\mathcal{D}(\mathbb{F}_2^t).$$
The second one is the uniform distribution over all vectors $z\in\mathbb{F}_2^N$ of weight $\textnormal{wt}(z)\in S$, for some $S\subseteq\{0,1,\dotsc,N\}$. We will denote this probability distribution by
$$\lambda_{S}:=\mathcal{D}(\{z\in\mathbb{F}_2^N:\textnormal{wt}(z)\in S\}).$$

\subsection{Probability theory}
We will need two very standard results of probability theory (see for e.g. \cite{theorybook}): Markov's inequality and Chernoff's bound. We start with Markov's inequality.
\begin{lemma}\label{markov}
Let $X$ be a non-negative random variable. Then for any $a>0$, we have
$$\Pr[X\geq a]\leq \frac{\E[X]}{a}.$$
\end{lemma}
 We will also need Chernoff's bound:
 \begin{lemma}\label{chernoff}
 Let $X_1,X_2,\dotsc,X_m$ be i.i.d. random variables taking values in $\{0,1\}$, and define $X:=X_1+X_2+\dotsc+X_m$. Then for any $\delta\in (0,1)$, we have
 $$\Pr\Big[\big|X-\E[X]\big|>\delta\cdot m \E[X_1]\Big]\leq 2e^{-\frac{\delta^2 \cdot m\E[X_1]}{3}}.$$
 \end{lemma}
\subsection{Fourier analysis}\label{sectionfourier}
The Fourier basis is a useful basis for the space of functions mapping $\F_2^N$ to the real numbers. We recall some of its properties below (see for e.g. \cite{2008wolffourier}). For $f,g \in \F_2^N\rightarrow \R$, define the inner product $$ \langle f, g \rangle := \frac{1}{2^N}\sum_{x \in \F_2^N} f(x) g(x).$$
For every $x,y \in \F_2^N$, define the character
$$ \chi_y(x) := (-1)^{\sum_{j=1}^N  x_j y_j }.$$ These functions form an orthonormal basis, namely for $y,y' \in \F_2^N$,
\begin{align*}
\langle \chi_y, \chi_{y'} \rangle = \begin{cases}
1 & \text{if $y=y'$,}\\
0 & \text{otherwise.}
\end{cases}
\end{align*}
We define the Fourier coefficients $\hat{f}(y) := \langle f, \chi_y \rangle$. Then for $f,g : \F_2^N \rightarrow \R$, we have
$$\langle f, g \rangle = \sum_{y \in \F_2^N} \hat{f}(y) \cdot \hat{g}(y).$$
In particular, $$ \frac{1}{2^N}\sum_{x\in\F_2^N}f(x)^2 = \sum_{y\in\F_2^N} \hat{f}(y)^2.$$

\section{Collisions vs decoding}\label{collisionsdecoding}
Recall that we denote by $P_\epsilon$ the $\epsilon$-Bernoulli distribution over $\mathbb{F}_2^N$, i.e. the distribution
$$P_\epsilon (z):=\epsilon^{\textnormal{wt}(z)}(1-\epsilon)^{N-\textnormal{wt}(z)}.$$
Recall also that for any subset $S\subseteq \{0,1,\dotsc,N\}$, we denote by $\lambda_{S}$ the uniform distribution over all strings $z\in\mathbb{F}_2^N$ of weight $\textnormal{wt}(z)\in S$, i.e.
$$\lambda_{S}(z):=\begin{cases}
\frac{1}{\sum_{j\in S} \binom{N}{j}} & \text{if $\textnormal{wt}(z)\in S$,}\\
0 & \text{otherwise.}
\end{cases}$$
The goal of this section will be to analyze the relationship between the decoding of an error string $\rho\in\mathbb{F}_2^N$ and the collision probability of strings $z\in\mathbb{F}_2^N$ within the map $z\mapsto Hz^\intercal$. Intuitively, the more collisions there are within this mapping, the harder it is for our decoder to correctly identify the error string $\rho$ upon seeing only its image $H\rho^\intercal$. However, certain error strings might be unlikely enough to occur that our decoder can safely ignore them. For example, if we are interested in an $\epsilon$-noisy error string $\rho$, then $\rho$ is unlikely to have weight $\textnormal{wt}(\rho)$ far away from $\epsilon N$. We could thus choose to ignore all strings whose weights do not lie in the set $S=[\epsilon N\pm l]\cap \N$, for some integer $l$. In order to analyze the collisions that occur when strings are required to have weight $\textnormal{wt}(z)\in S$, we define for every $z\in\mathbb{F}_2^N$ and every $S\subseteq \{0,1,\dotsc,N\}$ the set of $S$-colliders of $z$, i.e. the set of strings $y$ that lie in the coset of $z$ and have weight $\textnormal{wt}(y)\in S$:
\begin{definition}
For any $z\in\mathbb{F}_2^N$, any matrix $H$ with $N$ columns and entries in $\F_2$, and any subset $S\subseteq\{0,1,\dotsc,N\}$, define
$$\Omega_z^{S,H}:=\left\{y\in \mathbb{F}_2^N :\textnormal{wt}(y)\in S\textnormal{ and } Hy^{\intercal}=Hz^\intercal   \right\}.$$
When $H$ is clear from context, we will drop the superscript and write $\Omega_z^{S}$.
\end{definition}
This definition captures a natural parameter for how large of a list we need before we can confidently claim that it contains the error string: if we are given $H\rho^\intercal$ and are told that with high probability the error string $\rho$ has weight $\textnormal{wt}(\rho)\in S$, then we should output the list $\Omega_\rho^S$. For unique decoding we want to argue that $|\Omega_\rho^S|=1$ with high probability, whereas for list decoding we want to argue that $|\Omega_\rho^S|\leq k$ with high probability, for some integer $k>1$. The expectation of $|\Omega_\rho^S|$ will thus be a key quantity in our analysis. We will call this expectation the "collision count."
\begin{definition}
For any subset $S\subseteq \{0,1,\dotsc,N\}$ and any matrix $H$ with $N$ columns and entries in $\mathbb{F}_2$, define
\begin{align*}
\mathsf{Coll}_H(S)&:=\E_{z\sim \lambda_S}\big[|\Omega_z^S|\big].
\end{align*}
\end{definition}
When the set $S$ only contains one or two elements (i.e. $S=\{w\}$ or $S=\{w,w'\}$), we will abuse notation and write $\mathsf{Coll}_H(w)$ and $\mathsf{Coll}_H(w,w')$ to mean $\mathsf{Coll}_H(\{w\})$ and $\mathsf{Coll}_H(\{w,w'\})$ respectively. In the following lemma, we use Markov's inequality to bound the probability of a list decoding error in terms of $\mathsf{Coll}_H(S)$.
\begin{lemma}\label{collisionuniform}
For any subset $S\subseteq \{0,1,\dotsc,N\}$, any matrix $H$ with $N$ columns and entries in $\F_2$, and any integer $k\geq 1$, we have
\begin{align*}
  \Pr_{\rho\sim \lambda_{S}}\big[|\Omega_\rho^S|>k\big]\leq\frac{ \mathsf{Coll_H}(S) -1}{k}.
\end{align*}
\end{lemma}
\begin{proof}
Note that $|\Omega_z^S|\geq 1$ for any $z\in\F_2^N$ with weight $\textnormal{wt}(z)\in S$, so the random variable $|\Omega_\rho^S|- 1$ is always non-negative. Applying Markov's inequality (i.e. Lemma \ref{markov}), we then have
\begin{align*}
  \Pr_{\rho\sim \lambda_{S}}\big[|\Omega_\rho^S|>k\big]&=\Pr_{\rho\sim \lambda_{S}}\big[|\Omega_\rho^S|-1\geq k\big]\\
  &\leq\frac{ \mathsf{Coll_H}(S) -1}{k}.
\end{align*}
\end{proof}
When the error string $\rho$ is sampled uniformly at random from the set $\{z\in\mathbb{F}_2^N:\textnormal{wt}(z)\in S\}$, the above lemma allows us to relate the decoding error probability to the collision count $\mathsf{Coll_H}(S)$. The problem we are most interested in, however, is when $\rho$ is sampled not from some uniform distribution, but from the $\epsilon$-noisy probability distribution $P_\epsilon$. We will now show how to connect these two decoding problems. The intuition is that
by the Chernoff bound, we only need to concern ourselves with strings whose weights lie in $S=[\epsilon N\pm l]\cap \N$, for some appropriately chosen $l$. But in this weight band all strings have similar weight, and so are given similar probability under the distribution $P_\epsilon$. Intuitively, the $P_\epsilon$-decoder must then perform very similarly to the $\lambda_S$-decoder. The following proposition makes this idea precise, and then uses Lemma \ref{collisionuniform} to bound the probability of a decoding error. Recall that $D_k:\F_2^t\rightarrow(\F_2^N)^{\otimes k}$ is the max-likelihood decoder $$D_k(x):=\mathop{\textnormal{argmin}}_{\substack{\{z_1,z_2\dotsc,z_k\}\subseteq \mathbb{F}_2^N\\Hz_i^\intercal=x^\intercal \textnormal{ for all }i}}\{\textnormal{wt}(z_1)+\textnormal{wt}(z_2)+\dotsc+\textnormal{wt}(z_k)\},$$
where ties are broken according to the lexicographic order.
\begin{prop}\label{collisiondecoding}
Let $H$ be any matrix with $N$ columns and entries in $\F_2$. Consider any noise parameter $\epsilon\in(0,\frac{1}{2})$ and any $ l\in[1,\min\{ \epsilon N,(\frac{1}{2}-\epsilon)N\}]$. Then
\begin{enumerate}[label=(\roman*)]
    \item We have the following unique-decoding bound.
    \begin{align*}
    \Pr_{\rho\sim P_\epsilon}[\rho\notin D_1(H\rho^\intercal)]&\leq 2e^{-\frac{l^2}{3\epsilon N}} +4(l+1) \mathop{\max}_{\substack{S\subseteq [\epsilon N\pm l]\cap\N \\
    1\leq|S|\leq 2}} \Big\{ \mathsf{Coll_H}(S)-1 \Big\}.
\end{align*}
    \item Consider some integer $k> 1$ satisfying $\frac{k}{2l+1}\in\N$. Then we have the following list-decoding bound for list size $k$.
\begin{align*}
    \Pr_{\rho\sim P_\epsilon}[\rho\notin D_k(H\rho^\intercal)]&\leq 2e^{-\frac{l^2}{3\epsilon N}} +\frac{4(l+1)}{k} \mathop{\max}_{\substack{w\in [\epsilon N\pm l]\cap\N }} \Big\{ \mathsf{Coll_H}(w)-1 \Big\}.
\end{align*}
\end{enumerate}
\end{prop}

\begin{proof}
We will consider the unique decoding case ($k=1$) and the list-decoding case ($k>1$) separately.
\hfil\\
\textbf{Case 1: Unique decoding, i.e. $k=1$}
\hfil\\
Let $t$ be the number of rows in the matrix $H$. We will show that a slightly less performant decoder $\tilde{D}_1:\mathbb{F}_2^t\rightarrow \mathbb{F}_2^N$ satisfies the desired probability bound. We define $\tilde{D}_{1}$ as follows: upon receiving input $x\in\mathbb{F}_2^t$, $\tilde{D}_1$ outputs the minimum-weight string from the set $\big\{z\in\F_2^N:Hz^\intercal=x^\intercal,\textnormal{wt}(z)\in[\epsilon N\pm l]\cap\N\big\}.$
If this set is empty, the decoder fails. If there are multiple minimal-weight strings in the set, the decoder outputs the first one in the lexicographic order. It is clear that
$$\Pr_{\rho\sim P_\epsilon}[ \rho\neq D_{1}(H\rho^\intercal)]\leq \Pr_{\rho\sim P_\epsilon}[ \rho\neq \tilde{D}_1(H\rho^\intercal)],$$
since $D_1$ always returns the most likely string whereas $\tilde{D}_1$ may not.
We thus turn to proving the desired bound for $\tilde{D}_{1}$. We first bound the probability that the error string $\textnormal{wt}(\rho)$ be far away from its mean. Letting
$$B=\{z\in\mathbb{F}_2^N:\big|\textnormal{wt}(z)-\epsilon N\big|\leq l\},$$
we have by Chernoff's bound (i.e. Lemma \ref{chernoff}) that
\begin{align}\label{chernoffonk1}
    \Pr_{\rho\sim P_\epsilon}[ \rho\neq \tilde{D}_{1}(H\rho^\intercal)]&\leq \Pr_{\rho\sim P_\epsilon}[\rho\notin B]+\Pr_{\rho\sim P_\epsilon}[ \rho\neq \tilde{D}_{1}(H\rho^\intercal)\big|\rho\in B]\nonumber\\
    &\leq 2e^{-\frac{l^2}{3\epsilon N}} +\Pr_{\rho\sim P_\epsilon}[ \rho\neq \tilde{D}_{1}(H\rho^\intercal)\big|\rho\in B].
\end{align}
We want to bound the second term. For any $\rho\in B$, we define the set of "problematic weights" $S(\rho):=\{\lceil\epsilon N-l\rceil,\lceil\epsilon N-l\rceil+1,\dotsc,\textnormal{wt}(\rho)\}$ . We note that for $\rho\in B$, our decoder $\tilde{D}_{1}$ can only fail if there is some string $z\neq \rho$ satisfying $Hz^\intercal=H\rho^\intercal$ and $\textnormal{wt}(z)\in S(\rho)$. Recalling the definition $\Omega_\rho^S:=\{z:Hz^\intercal=H\rho^\intercal, \textnormal{wt}(z)\in S\}$, we can then rewrite our equation (\ref{chernoffonk1}) as
\begin{align*}
    \Pr_{\rho\sim P_\epsilon}[ \rho\neq \tilde{D}_{1}(H\rho^\intercal)]&\leq 2e^{-\frac{l^2}{3\epsilon N}}+\Pr_{\rho\sim P_\epsilon}\big[ |\Omega_\rho^{S(\rho)}|>1 \big|\rho\in B \big].
\end{align*}
Considering the most problematic weight level $w$ within the region $[\epsilon N\pm l]\cap\N$ and using a union bound over all lower levels $w'\leq w$, we get
\begin{align*}
    \Pr_{\rho\sim P_\epsilon}[ \rho\neq \tilde{D}_{1}(H\rho^\intercal)]&\leq 2e^{-\frac{l^2}{3\epsilon N}}+\max_{w\in[\epsilon N\pm l]\cap\N}\left\{\Pr_{\rho\sim P_\epsilon}\big[ |\Omega_\rho^{S(\rho)}|>1 \big|\textnormal{wt}(\rho)=w \big]\right\}\\
    &\leq 2e^{-\frac{l^2}{3\epsilon N}}\\
&\quad+(2l+1)\mathop{\max}_{\substack{w,w'\in[\epsilon N\pm l]\cap\N\\w'\leq w}}\left\{\Pr_{\rho\sim P_\epsilon}\big[ |\Omega_\rho^{\{w,w'\}}|>1\big|\textnormal{wt}(\rho)=w \big]\right\}.
\end{align*}
We now note that under the condition $\textnormal{wt}(\rho)=w$, the $\epsilon$-noisy probability distribution $P_\epsilon(\rho)$ and the uniform probability distribution $\lambda_{\{w,w'\}}(\rho)$ are identical (they are both uniform on strings of weight $w$). We can thus rewrite our last inequality as
\begin{align*}
    \Pr_{\rho\sim P_\epsilon}[ \rho\neq \tilde{D}_{1}(H\rho^\intercal)]    & \leq 2e^{-\frac{l^2}{3\epsilon N}}\\
& +(2l+1)\mathop{\max}_{\substack{w,w'\in[\epsilon N\pm l]\cap\N\\w'\leq w}}\left\{\Pr_{\rho\sim \lambda_{\{w,w'\}}}\big[ |\Omega_\rho^{\{w,w'\}}|>1\big|\textnormal{wt}(\rho)=w \big]\right\}.
\end{align*}
But by basic conditional probability we know that
{\small $$\Pr_{\rho\sim \lambda_{\{w,w'\}}}\big[ |\Omega_\rho^{\{w,w'\}}|>1 \big]\geq \Pr_{\rho\sim \lambda_{\{w,w'\}}}\big[ \textnormal{wt}(\rho)=w \big] \Pr_{\rho\sim \lambda_{\{w,w'\}}}\big[ |\Omega_\rho^{\{w,w'\}}|>1\big|\textnormal{wt}(\rho)=w \big],$$}
so we can bound our previous expression by
\begin{equation}\label{useconditional}
\begin{split}
    \Pr_{\rho\sim P_\epsilon}[ \rho\neq \tilde{D}_{1}(H\rho^\intercal)]    &\leq 2e^{-\frac{l^2}{3\epsilon N}}\\
&\quad+(2l+1)\mathop{\max}_{\substack{w,w'\in[\epsilon N\pm l]\cap\N\\w'\leq w}}\left\{\frac{\Pr_{\rho\sim \lambda_{\{w,w'\}}}\big[ |\Omega_\rho^{\{w,w'\}}|>1 \big]}{\Pr_{\rho\sim \lambda_{\{w,w'\}}}\big[ \textnormal{wt}(\rho)=w \big]}\right\}.
\end{split}
\end{equation}
Now, from our theorem's assumption on $l$, we know that any $w,w'\in [\epsilon N\pm l]\cap\N$ must lie in the interval $[0,\frac{N}{2}]$. Combining this with the fact that $w'\leq w$, we have
\begin{align}\label{boundbyhalf}
\Pr_{\rho\sim \lambda_{\{w,w'\}}}\big[ \textnormal{wt}(\rho)=w \big]=\frac{\binom{N}{w}}{\binom{N}{\{w,w'\}}}\geq\frac{\binom{N}{w}}{\binom{N}{w}+\binom{N}{w'}}\geq \frac{1}{2}.
\end{align}
We note that the inequality above holds for both the case $w\neq w'$ and the case $w=w'$. (When $w=w'$, we have $\binom{N}{\{w,w'\}}=\binom{N}{w}\leq \binom{N}{w}+\binom{N}{w'}$). It then follows from (\ref{useconditional}) and (\ref{boundbyhalf}) that
\begin{align*}
    \Pr_{\rho\sim P_\epsilon}\big[ \rho\notin \tilde{D}_1(H\rho^\intercal) \big]\leq  2e^{-\frac{l^2}{3\epsilon N}}+      2(2l+1)\cdot \mathop{\max}_{\substack{S\subseteq [\epsilon N\pm l]\cap\N \\
    |S|\in\{1,2\}}} \left\{\Pr_{\rho\sim \lambda_{S}}\big[|\Omega_\rho^{S}|>1  \big]\right\}.
\end{align*}
The theorem statement then follows from Lemma \ref{collisionuniform}.

\newpage
\hfil\\
\textbf{Case 2: List decoding, i.e. $k>1$}
\hfil\\
Let $t$ be the number of rows in the matrix $H$. We will show that a slightly less performant decoding function $D_{k,l}:\mathbb{F}_2^t\rightarrow (\mathbb{F}_2^N)^{\otimes k}$ satisfies the desired probability bound. We define $D_{k,l}$ as follows: upon receiving input $x\in\mathbb{F}_2^t$, $D_{k,l}$ outputs $\frac{k}{2l+1}$ strings from $\{z\in\mathbb{F}_2^N:Hz=x, \textnormal{wt}(z)=w\}$, for each $w\in[\epsilon N\pm l]\cap\N$. If there are fewer than $\frac{k}{2l+1}$ strings in some level $w$, the decoder returns all of them. If there are more than $\frac{k}{2l+1}$ strings in some level $w$, the decoder returns the first $\frac{k}{2l+1}$ ones in lexicographic order. It is clear that for any $l$ we have
$$\Pr_{\rho\sim P_\epsilon}[ \rho\notin D_{k}(H\rho^\intercal)]\leq \Pr_{\rho\sim P_\epsilon}[ \rho\notin D_{k,l}(H\rho^\intercal)],$$
since $D_k$ returns the $k$ most likely strings while $D_{k,l}$ returns at most $k$ strings. We thus turn to proving the desired bound for $D_{k,l}$. Letting
$$B=\Big\{z\in\mathbb{F}_2^N:\big|\textnormal{wt}(z)-\epsilon N\big|\leq l\Big\},$$
we have by Chernoff's bound (i.e. Lemma \ref{chernoff}) that
\begin{align*}
    \Pr_{\rho\sim P_\epsilon}[ \rho\notin D_{k,l}(H\rho^\intercal)]&\leq \Pr_{\rho\sim P_\epsilon}[\rho\notin B]+\Pr_{\rho\sim P_\epsilon}[ \rho\notin D_{k,l}(H\rho^\intercal)\big|\rho\in B]\\
    &\leq 2e^{-\frac{l^2}{3\epsilon N}}+\max_{w\in[\epsilon N\pm l]\cap\N}\left\{\Pr_{\rho\sim P_\epsilon}[ \rho\notin D_{k,l}(H\rho^\intercal)\big|\textnormal{wt}(\rho)=w]\right\}.
\end{align*}
Since the distribution $P_\epsilon$ gives the same probability to any two strings of equal weights, we get
\begin{align*}
    \Pr_{\rho\sim P_\epsilon}[ \rho\notin D_{k,l}(H\rho^\intercal)] &\leq 2e^{-\frac{l^2}{3\epsilon N}}+\max_{w\in[\epsilon N\pm l]\cap\N}\left\{\Pr_{\rho\sim \lambda_{\{w\}}}[ \rho\notin D_{k,l}(H\rho^\intercal)]\right\}\\
    &\leq 2e^{-\frac{l^2}{3\epsilon N}}+\max_{w\in[\epsilon N\pm l]\cap\N}\left\{\Pr_{\rho\sim \lambda_{\{w\}}}\big[|\Omega_\rho^{\{w\}}|>\frac{k}{2l+1}\big]\right\}.
\end{align*}
Applying Lemma \ref{collisionuniform}, we get
\begin{align*}
    \Pr_{\rho\sim P_\epsilon}[ \rho\notin D_{k,l}(H\rho^\intercal)] &\leq 2e^{-\frac{l^2}{3\epsilon N}}+\frac{2l+1}{k}\cdot \max_{w\in[\epsilon N\pm l]\cap\N} \Big\{ \mathsf{Coll_H}(w)-1 \Big\}.
\end{align*}
\end{proof}

\section{A criterion for decoding}\label{weightdecoding}
In this section, we give a criterion that certifies that a linear code $C\subseteq\F_2^N$ is resilient to errors of probability $\epsilon$. We give such a criterion for both unique decoding and list decoding. The function we will need to make this connection is the Krawtchouk polynomial of degree $s$, which is defined as
$$K_s(x):=\sum_{j=0}^s(-1)^j\binom{x}{j}\binom{N-x}{s-j},$$
where for any polynomial $p(x)$ we defined $\binom{p(x)}{j}:=\frac{p(x)(p(x)-1)\dotsc (p(x)-j+1)}{j!}$.
For vectors $v\in\F_2^N$, we will abuse notation and write $K_s(v)$ to mean $K_s(\textnormal{wt}(v)).$ For convenience, we also define for any $S\subseteq\{0,1,\dotsc,N\}$ the function
$$K_S(x):=\sum_{s\in S}K_s(x).$$
In the following proposition, we use basic Fourier analysis tools to rewrite the collision count $\mathsf{Coll}_H(S)$ in terms of the Krawtchouk polynomial $K_S$. We note that Proposition \ref{collisionprob} was previously proven in a different form in \cite{2021bargfourierlink} (see Theorem 2.1 and Lemma 4.1), and can be seen as describing the coset weight distribution of the code. Recall that we use $\mu_t$ to denote the uniform distribution over all vectors in $\mathbb{F}_2^t$, and that we use the notation $\binom{N}{S}:=\sum_{s\in S}\binom{N}{s}$.
\begin{prop}\label{collisionprob}
Fix $\epsilon\in(0,\frac{1}{2})$, and let $H$ be a $t\times N$ matrix with entries in $\mathbb{F}_2$. Then for any $S\subseteq\{0,1,\dotsc,N\}$, we have
\begin{align*}
\mathsf{Coll}_H(S)&=\frac{1}{\binom{N}{S}}\E_{v\sim\mu_t}[K_S(vH)^2].
\end{align*}
\end{prop}

\begin{proof}
The main tool we will use is Parseval's Identity, which relates the evaluations $f(x)$ of a function $f:\mathbb{F}_2^t\rightarrow \R$ to its Fourier coefficients $\hat{f}(y)$ by
\begin{align}\label{parseval}
    \frac{1}{2^t}\sum_{x\in\mathbb{F}_2^t} f(x)^2=\sum_{y\in\mathbb{F}_2^t}\hat{f}(y)^2.
\end{align}
We will first need to rewrite $\mathsf{Coll}_H(S)$ as the $\ell_2$ norm of some function $f$. For this, we recall the definition $\Omega_z^S:=\left\{y\in \mathbb{F}_2^N :\textnormal{wt}(y)\in S\textnormal{ and } Hy^{\intercal}=Hz^\intercal   \right\}$ and note that
\begin{align*}
\mathsf{Coll}_H(S)&:=\frac{1}{\binom{N}{S}}\sum_{z\in\F_2^N:\textnormal{wt}(z)\in S}|\Omega_z^S|\\
&= \binom{N}{S}\sum_{z\in\F_2^N:\textnormal{wt}(z)\in S}\frac{1}{|\Omega_z^S|}\Pr_{a\sim\lambda_S}[Ha^\intercal=Hz^\intercal]^2\\
&=\binom{N}{S}\sum_{x \in \F_2^{t}} \Pr_{z\sim\lambda_S}[Hz^\intercal=x^\intercal]^2.
\end{align*}
We are now ready to apply Parseval's Identity. Letting $f(x) =  \Pr_{z\sim \lambda_S}[Hz^\intercal = x^\intercal]$ in equation (\ref{parseval}), we get
\begin{align*}
\mathsf{Coll}_H(S)& = \binom{N}{S}\sum_{x \in \F_2^{t}} f(x)^2 \\
& = 2^t\binom{N}{S}\sum_{y \in \F_2^{t}} \hat{f}(y)^2.
\end{align*}
But by definition of the Fourier transform, we have $$\hat{f}(y):=2^{-t}\sum_{x \in \F_2^{t}} \frac{1}{\binom{N}{S}}\big|\{z\in\F_2^N:\textnormal{wt}(z)\in S\textnormal{ and }Hz^\intercal=x^\intercal\}\big|  \cdot(-1)^{y\cdot x^\intercal},$$ so our previous equation can be rewritten as
\begin{align}\label{expand1}
&\mathsf{Coll}_H(S) \nonumber\\
&\quad=2^{t}\binom{N}{S} \sum_{y \in \F_2^{t}} \Big (2^{-t} \sum_{x \in \F_2^{t}} \frac{1}{\binom{N}{S}}(-1)^{y\cdot x^\intercal }\cdot\big|\{z\in\F_2^N:\textnormal{wt}(z)\in S\textnormal{ and }Hz^\intercal=x^\intercal\}\big|     \Big)^2\nonumber\\
&\quad= 2^{-t}\frac{1}{\binom{N}{S}} \sum_{y \in \F_2^{t}} \Big ( \sum_{\substack{z\in\mathbb{F}_2^N \\ \textnormal{wt}(z)\in S}}   (-1)^{y\cdot Hz^\intercal } \Big)^2.
\end{align}
We now note that by definition, for any non-negative integer $s\leq N$ we have
\begin{align*}
    K_s(yH)&:=\sum_{j=0}^s(-1)^j\binom{\textnormal{wt}(yH)}{j}\binom{N-\textnormal{wt}(yH)}{s-j}\\
    &=\sum_{\substack{z\in\mathbb{F}_2^N \\ \textnormal{wt}(z)=s}}  (-1)^{yH\cdot z^\intercal},
\end{align*}
where we used the convention that $\binom{a}{b}=0$ when $a<b$. Combining this with equation (\ref{expand1}), we get
\begin{align*}
\mathsf{Coll}_H(S)&=\frac{2^{-t}}{\binom{N}{S}} \sum_{y \in \F_2^{t}} K_S(yH)^2.
\end{align*}
\end{proof}
We will now combine Propositions \ref{collisiondecoding} and \ref{collisionprob} to obtain Theorem \ref{fouriercriterionunique}, i.e. to obtain a bound on the decoding error probability in terms of Krawtchouk polynomials. We prove a generalized version of Theorem \ref{fouriercriterionunique} below. To recover Theorem \ref{fouriercriterionunique}, set the list size $k=1$ and set $l= N^{3/4} $, and apply Lemma \ref{eqvltdecoder}.
(You want to think of the parameter $l$ as being $l>>\sqrt{N}$ in both the case $k=1$ and the case $k>1$, so that the error term $e^{-\frac{l^2}{3\epsilon N}}$ is small).
\begin{theorem}\label{fouriercriteriongeneral}
Let $H$ be any $t\times N$ matrix with entries in $\mathbb{F}_2$. Consider any noise parameter $\epsilon\in(0,\frac{1}{2})$ and any $ l\in[1,\min\{ \epsilon N,(\frac{1}{2}-\epsilon)N\}]$. Then
\begin{enumerate}[label=(\roman*)]
    \item We have the following unique-decoding bound.
    \begin{align*}
    \Pr_{\rho\sim P_\epsilon}[\rho\notin D_1(H\rho^\intercal)]&\leq 2e^{-\frac{l^2}{3\epsilon N}} +4(l+1) \mathop{\max}_{\substack{S\subseteq [\epsilon N\pm l]\cap\N \\
    1\leq|S|\leq 2}} \Big\{  \frac{1}{\binom{N}{S}}\E_{v\sim\mu_t}\big[ K_S(vH)^2\big]-1 \Big\}.
\end{align*}
    \item Consider some integer $k> 1$ satisfying $\frac{k}{2l+1}\in\N$. Then we have the following list-decoding bound for list size $k$.
\begin{align*}
    \Pr_{\rho\sim P_\epsilon}[\rho\notin D_k(H\rho^\intercal)]&\leq 2e^{-\frac{l^2}{3\epsilon N}} +\frac{4(l+1)}{k} \mathop{\max}_{\substack{w\in [\epsilon N\pm l]\cap\N }} \Big\{  \frac{1}{\binom{N}{w}}\E_{v\sim\mu_t}\big[ K_w(vH)^2\big]-1 \Big\}.
\end{align*}
\end{enumerate}
\end{theorem}

\begin{proof}
The theorem statement follows directly from Propositions \ref{collisiondecoding} and \ref{collisionprob}.
\end{proof}
One consequence of Theorem \ref{fouriercriteriongeneral} is Corollary \ref{weightunique}, which states that $C$ is resilient to $\epsilon$-errors if the weight distribution of $C^\perp$ is close enough to the binomial distribution (see Appendix \ref{abinomialweight} for the proof).
As another application of Theorem \ref{fouriercriteriongeneral}, we present the following bound on the probability of making a list-decoding error for a code $C$. We note that once again, our bound depends only on the weight distribution of the dual code $C^\perp$.

\begin{prop}\label{weightcriterion}
Fix any $\epsilon\in(0,\frac{1}{2})$, and define
$\beta:=\frac{1-2\sqrt{\tilde{\epsilon}(1-\tilde{\epsilon})}}{2}$ for
$\tilde{\epsilon}=\epsilon+\frac{1}{\sqrt{\log N}}.$
Let $B=[\beta N,(1-\beta)N]\cap\N $, and let $k^*=(2\lfloor \frac{N}{\sqrt{\log N}}\rfloor +1)m$ for some integer $m> 0$.
Then for any integer $N> 2^{\frac{1}{\epsilon^2(1-\epsilon)^2}+1} $ and all list sizes $k\geq k^*$, we have that any $t\times N$ matrix $H$ with entries in $\mathbb{F}_2$ satisfies
\begin{align*}
\Pr_{\rho\sim P_\epsilon}[\rho\notin D_k(H\rho^\intercal)]&\leq 2e^{-\frac{N}{4\epsilon\log N}}+\frac{N}{k^*} \max_{j\in B}\left\{ \Pr_{v\sim \mu_t}[\textnormal{wt}(vH)=j] \cdot \frac{2^N}{\binom{N}{j}}-1\right\} \nonumber\\
&+\frac{2^{h(\epsilon)N+5h(\frac{1}{\sqrt{\log N}})N}}{k^*}  \max_{ j\notin B}\left\{ \Pr_{v\sim \mu_t}\big[\textnormal{wt}(vH)=j\big]\cdot  2^{2\epsilon N\log|1-\frac{2j}{N}|} \right\}.
\end{align*}
\end{prop}

\begin{proof}
We will use Theorem \ref{fouriercriteriongeneral} to bound the decoding error probability in terms of the Krawtchouk polynomials $K_S(j)$ and the probability factors $\Pr_{v\sim\mu_{t}}\big[\textnormal{wt}(vH)=j\big] $. Some of these terms will then be bounded using Proposition \ref{IFourier}, and some will be bounded using Theorem \ref{thmkrawtchoukbound}.
We proceed with the proof; applying Theorem \ref{fouriercriteriongeneral} to the list size $k^*$ with parameter $l=\lfloor \frac{N}{\sqrt{\log N}}\rfloor$, we get
\begin{equation}\label{startingpoint}
\begin{split}
    \Pr_{\rho\sim P_\epsilon}[\rho\notin D_k(H\rho^\intercal)]&\leq \Pr_{\rho\sim P_\epsilon}[\rho\notin D_{k^*}(H\rho^\intercal)]\\
    &\leq 2e^{-\frac{N}{4\epsilon\log N}}\\
&+\frac{N}{k^*}  \mathop{\max}_{\substack{w\in [\epsilon N\pm \frac{N}{\sqrt{\log N}}]\cap\N }} \Big\{  \frac{1}{\binom{N}{w}}\sum_{j=0}^N \Pr_{v\sim\mu_{t}}\big[\textnormal{wt}(vH)=j\big] K_w(j)^2 -1 \Big\}.
\end{split}
\end{equation}
We want to bound the summation in the second term. We will start with the central terms $j\in B$. For these we rely on Proposition \ref{IFourier}, which states that $\frac{2^{-N}}{\binom{N}{w}}\sum_{j = 0}^N  \binom{N}{j} \cdot K_w(j)^2=1$ for all $w\in\{0,1,\dotsc,N\}$. For any $w\in\{0,1,\dotsc,N\}$, we thus get
\begin{align}\label{center1}
&\frac{1}{\binom{N}{w}}\sum_{j\in B} \Pr_{v\sim\mu_{t}}\big[\textnormal{wt}(vH)=j\big] K_w(j)^2\nonumber\\
&\quad\quad\quad\quad\quad\quad\quad\quad\leq \frac{1}{\binom{N}{w}}\max_{j\in B}\left\{ \Pr_{v\sim \mu_t}[\textnormal{wt}(vH)=j]\cdot \frac{1}{\binom{N}{j}}  \right\}\sum_{j\in B} \binom{N}{j} \cdot K_w(j)^2\nonumber \\
&\quad\quad\quad\quad\quad\quad\quad\quad\leq 2^{N}\max_{j\in B}\left\{ \Pr_{v\sim \mu_t}[\textnormal{wt}(vH)=j] \cdot \frac{1}{\binom{N}{j}}\right\}.
\end{align}
We then want to bound the contribution of the faraway terms $j\notin B$ to the summation in (\ref{startingpoint}), i.e. we want to bound
\begin{equation}\label{bdgoal}
    \mathop{\max}_{\substack{w\in [\epsilon N\pm \frac{N}{\sqrt{\log N}}]\cap\N }} \Big\{  \frac{1}{\binom{N}{w}}\sum_{j\notin B} \Pr_{v\sim\mu_{t}}\big[\textnormal{wt}(vH)=j\big] K_w(j)^2  \Big\}. \tag{\mbox{ $*$ }}
\end{equation}
Bounding this quantity by $N$ times its maximum value over $j$ and applying Theorem \ref{thmkrawtchoukbound}, we get
\begin{align*}
    (\ref{bdgoal})&\leq\frac{N}{\binom{N}{\lceil\epsilon N-\frac{N}{\sqrt{\log N}}\rceil}} \mathop{\max}_{\substack{w\in[\epsilon N\pm \frac{N}{\sqrt{\log N}}]\cap\N\\
    j\notin B}}\Big\{ \Pr_{v\sim\mu_{t}}\big[\textnormal{wt}(vH)=j\big]K_{w}(j)^2  \Big\}\\
    &\leq \frac{ N }{\binom{N}{\lceil\epsilon N-\frac{N}{\sqrt{\log N}}\rceil}}\max_{\substack{w\in [\epsilon N\pm \frac{N}{\sqrt{\log N}}]\cap \N \\j\notin B}}\left\{ \Pr_{v\sim \mu_t}\big[\textnormal{wt}(vH)=j\big]\cdot 2^{2h(\frac{w}{N})N+2w\log|1-\frac{2j}{N}|} \right\}.
\end{align*}
But by Lemma \ref{stirling} and subadditivity of entropy (i.e. Lemma \ref{entropysubadditive}), we know that \begin{align*}
\binom{N}{\lceil\epsilon N-\frac{N}{\sqrt{\log N}}\rceil}\geq \frac{1}{\sqrt{2N}}2^{h(\epsilon-\frac{1}{\sqrt{\log N}})N}\geq\frac{1}{\sqrt{2N}}2^{h(\epsilon)N-h(\frac{1}{\sqrt{\log N}})N}.
\end{align*}Additionally, for any $w\in\{\epsilon N\pm \frac{N}{\sqrt{\log N}}\}$ we have (again by subadditivity of entropy, i.e. Lemma \ref{entropysubadditive})
\begin{align*}
2h(\frac{w}{N})N\leq2h(\epsilon+\frac{1}{\sqrt{\log N}})N\leq 2h(\epsilon)N+2h(\frac{1}{\sqrt{\log N}})N.
\end{align*}

Finally, for any $w\in\{\epsilon N\pm \frac{N}{\sqrt{\log N}}\}$ and any $j\notin B$, we have $2w\log|1-\frac{2j}{N}|\leq 2\epsilon N \log|1-\frac{2j}{N}|-2\frac{N}{\sqrt{\log N}}\log|1-2\beta|\leq 2\epsilon N \log|1-\frac{2j}{N}|+h(\frac{1}{\sqrt{\log N}})N$, where the last inequality follows from our assumption that $N> 2^{\frac{1}{\epsilon^2(1-\epsilon)^2}+1} $. Overall, we then get
\begin{align*}
(\ref{bdgoal})&\leq \sqrt{2}N^{\frac{3}{2}}\cdot 2^{4h(\frac{1}{\sqrt{\log N}})N}\cdot2^{h(\epsilon)N}  \max_{ j\notin B}\left\{ \Pr_{v\sim \mu_t}\big[\textnormal{wt}(vH)=j\big]\cdot  2^{2\epsilon N\log|1-\frac{2j}{N}|} \right\} \\
&\leq \frac{1}{N}\cdot 2^{5h(\frac{1}{\sqrt{\log N}})N}\cdot2^{h(\epsilon)N}  \max_{ j\notin B}\left\{ \Pr_{v\sim \mu_t}\big[\textnormal{wt}(vH)=j\big]\cdot  2^{2\epsilon N\log|1-\frac{2j}{N}|} \right\},
\end{align*}
where the last line follows from our assumption that $N>2^{17}>50$ and the fact that for all $N>50$, we have $\log(\sqrt{2}N^{\frac{5}{2}})\leq 3\log N\leq \frac{N}{\sqrt{\log N}}\leq h(\frac{1}{\sqrt{\log N}})N.$
Combining this bound for the faraway terms with our bound (\ref{center1}) for the central terms of the summation, we bound the right-hand side of equation (\ref{startingpoint}) by
\begin{align*}
\Pr_{\rho\sim P_\epsilon}[\rho\notin D_k(H\rho^\intercal)]&\leq 2e^{-\frac{N}{4\epsilon\log N}}+\frac{N}{k^*} \max_{j\in B}\left\{ \Pr_{v\sim \mu_t}[\textnormal{wt}(vH)=j] \cdot \frac{2^N}{\binom{N}{j}}-1\right\} \nonumber\\
&+\frac{2^{h(\epsilon)N+5h(\frac{1}{\sqrt{\log N}})N}}{k^*}  \max_{ j\notin B}\left\{ \Pr_{v\sim \mu_t}\big[\textnormal{wt}(vH)=j\big]\cdot  2^{2\epsilon N\log|1-\frac{2j}{N}|} \right\}.
\end{align*}
\end{proof}

\section{The weight distribution of transitive linear codes}\label{weighttransitive}
We will now prove Theorem \ref{probtransitive}. We note that the bound we get is essentially tight, since for any finite field $\F_q$ and any integer divider $j$ of $N$, the repetition code $$C=\left\{(z,z,\dotsc,z)\in\F_q^N:z\in\F_q^{j}\right\}$$ is transitive, has dimension $j$, and has weight distribution
\begin{align*}
\Pr_{c\sim \mathcal{D}(C)}\Big[\textnormal{wt}(c) = \alpha N\Big] &=q^{-j} \cdot \binom{j}{(1-\alpha) j} (q-1)^{\alpha j}\\
&\geq q^{-j} \cdot \sqrt{\frac{1}{2j}}\cdot 2^{h(\alpha)j}\cdot q^{\alpha j \log_q (q-1)}\\
&= \sqrt{\frac{1}{2j}}\cdot q^{-(1-h_q(\alpha))j}
\end{align*}
for all $\alpha\in(0,1)$ such that $\alpha j\in\N$. We recall and prove our Theorem \ref{probtransitive} below:

\newtheorem*{probtransitive}{Theorem \ref{probtransitive}}
\begin{probtransitive}
Consider any finite field $\F_q$, and let $C\subseteq \mathbb{F}_q^N$ be any transitive linear code. Then for any $\alpha\in (0,1)$, we have
$$ \Pr_{c\sim \mathcal{D}(C)}\Big[\textnormal{wt}(c) = \alpha N\Big] \leq q^{-(1-h_q(\alpha)) \textnormal{dim }C},$$
where $\mathcal{D}(C)$ is the uniform distribution over all codewords in $C$, $\textnormal{wt}(c)$ is the number of non-zero coordinates of $c$, and $h_q$ is the q-ary entropy $$h_q(\alpha):= (1-\alpha) \log_q \frac{1}{1-\alpha} + \alpha \log_q\frac{q-1}{\alpha}.$$
\end{probtransitive}
\begin{proof}
Let $r=\textnormal{dim }C$, and let $M$ the $r\times N$ generator matrix of $C$. Without loss of generality, suppose that the first $r$ columns of $M$ span the column-space of $M$. Define
$$C^{(\alpha)}:=\{c\in C: \textnormal{wt}(c)=\alpha N\},$$
and let $Z=(Z_1,Z_2,\dotsc,Z_N)$ be a uniformly random codeword in $C^{(\alpha)}$. Now $C$ is transitive, so for every $j,k\in\{1,2,\dotsc,N\}$ the random variables $Z_j$ and $Z_k$ are identically distributed. By linearity of expectation and by definition of $C^{(\alpha)}$, we thus have that for every $j\in \{1,2,\dotsc,N\}$,
\begin{align}\label{prob0}
\Pr_{Z\sim \mathcal{D}(C^{(\alpha)})}[Z_j= 0]=1-\alpha.
\end{align}
Now for any nonzero $a,b\in\F_q$, there must be as many codewords $c\in C_\alpha$ with $c_j=a$ as there are codewords $c'\in C_\alpha$ with $c'_j=b$ (because $C$ is a linear subspace, so the mapping $c\mapsto ba^{-1}\cdot c$ maps codewords to codewords). The entropy of $Z_j$ can thus be expressed as
\begin{align}\label{eqentropy}
\mathop{\mathsf{H}}_{Z\sim \mathcal{D}(C^{(\alpha)})}(Z_j)&= (1-\alpha)\log\frac{1}{1-\alpha}+(q-1)\cdot\frac{\alpha}{q-1}\log\frac{q-1}{\alpha}\nonumber\\
&=h_q(\alpha)\log(q).
\end{align}
We will now show that $\mathsf{H}(Z_j |Z_1,Z_2,\dotsc,Z_{j-1})=0$ for every $j>r$. To this end, fix some $j>r$. Recall that the columns $\{M_1,M_2,\dotsc,M_r \}$ span the column-space of $M$, so we can write the column $M_j$ as $M_j=\sum_{k=1}^r \beta_k M_k$ for some $\beta_1,\beta_2\dotsc,\beta_r\in\F$. But any codeword $c\in C$ can be expressed as $v^{(c)}M$ for some $v^{(c)}\in\mathbb{F}
^r$, so any codeword $c\in C$ satisfies $$c_j=v^{(c)}M_j=\sum_{k=1}^r\beta_k v^{(c)}M_k=\sum_{k=1}^r\beta_k c_k.$$
The random variable $Z_j$ is thus determined by $\{Z_1,Z_2,\dotsc,Z_r \}$, and so we indeed have
$$\mathop{\mathsf{H}}_{Z\sim \mathcal{D}(C^{(\alpha)})}(Z_j|Z_1,Z_2,\dotsc,Z_{j-1})=0$$
for every $j>r$. Applying (\ref{eqentropy}) and the chain rule for entropy then gives
\begin{align*}
    \mathsf{H}(Z)&=\mathsf{H}(Z_1)+\sum_{i=2}^N\mathsf{H}(Z_i |Z_1,Z_2,\dotsc,Z_{i-1})    \\
    &\leq \sum_{i=1}^{r}\mathsf{H}(Z_i)  \\
    &\leq r \cdot h_q(\alpha)\log(q)
\end{align*}
Now $Z$ is sampled uniformly from $C^{(\alpha)}$, so $\mathsf{H}(Z)=\log \Big( |C^{(\alpha)}|\Big)$. We thus have
\begin{align*}
 \Pr_{c\sim \mathcal{D}(C)}\Big[\textnormal{wt}(c) = \alpha N\Big] &=\frac{\left|C^{(\alpha)}\right|}{q^r} \\
&=2^{\mathsf{H}(Z)}\cdot q^{-r}\\
&\leq q^{-(1-h_q(\alpha))\cdot r }.
\end{align*}
\end{proof}

\section{List decoding for transitive codes}\label{listtransitive}

We now turn to proving Theorem \ref{thmtransitivelist}. In section \ref{weightdecoding}, we bounded the minimum size for the decoding list of a linear code in terms of the weight distribution of its dual code. But as we stated in Claim \ref{dualtransitive}, the dual code of a transitive code is also transitive. For any transitive linear code $C$, we can thus apply our Theorem \ref{probtransitive} for the weight distribution of $C^\perp$ to get a bound on the size of the decoding list for $C$. We restate and prove our Theorem \ref{thmtransitivelist} below.

\newtheorem*{thmtransitivelist}{Theorem \ref{thmtransitivelist}}
\begin{thmtransitivelist}
Fix any $\epsilon\in(0,\frac{1}{2})$ and $\eta\in(0,1)$.
Then any  transitive linear code $C\subseteq\F_2^N$ of dimension $\textnormal{dim }C=\eta N$ can with high probability list-decode $\epsilon$-errors using a list $T$ of size
$$|T|=2^{\epsilon N\log(\frac{2}{1-\eta})+o(N)}+2^{4\epsilon N+o(N)}.$$
\end{thmtransitivelist}

\begin{proof}
We will show that for all $N> 2^{\frac{1}{\epsilon^2(1-\epsilon)^2}+1} $, there exists a function $T$ mapping every $x\in\F_2^N$ to a subset $T(x)\subseteq C$ of size
     $$|T(x)|=e^{\frac{N}{4\epsilon\log N}}\cdot 2^{5h(\frac{1}{\sqrt{\log N}})N}\cdot(2^{4\epsilon \eta N}+2^{\epsilon N\log(\frac{2}{1-\eta})}),$$
     with the property that for every codeword $c\in C$ we have
    \begin{align*}
    \Pr_{\rho\sim P_\epsilon}\big[c\notin T(c+\rho)\big]\leq4e^{-\frac{N}{4\epsilon\log N}}.
    \end{align*}
Let $H$ denote the parity-check matrix of $C$. By Lemma \ref{eqvltdecoder}, it is sufficient to show that for any list size $k>N$, we have
\begin{align}\label{eqvlttransitive}
    \Pr_{\rho\sim P_\epsilon}[\rho\notin  D_k(H\rho^\intercal )]&\leq 2e^{-\frac{N}{4\epsilon\log N}}+\frac{2^{5h(\frac{1}{\sqrt{\log N}})N+1}}{k} \cdot (2^{4\epsilon \eta N}+2^{\epsilon N\log(\frac{2}{1-\eta})}).
    \end{align}
Setting the list size $k=e^{\frac{N}{4\epsilon\log N}}\cdot 2^{5h(\frac{1}{\sqrt{\log N}})N}\cdot(2^{4\epsilon \eta N}+2^{\epsilon N\log(\frac{2}{1-\eta})})$ in equation (\ref{eqvlttransitive}) will then recover our theorem statement. We thus turn to proving (\ref{eqvlttransitive}). We note that $2\lfloor\frac{N}{\sqrt{\log N}}\rfloor+1<\frac{k}{2}$
%for $N> 2^{\frac{1}{\epsilon^2(1-\epsilon)^2}} $
, so there exists some $k^*\in[\frac{k}{2},k]$ satisfying the conditions of Proposition \ref{weightcriterion}. Proposition \ref{weightcriterion} then yields the following bound on the left-hand side of (\ref{eqvlttransitive}):
\begin{align}\label{thmbdtransitive}
\Pr_{\rho\sim P_\epsilon}[\rho\notin D_k(H\rho^\intercal)]&\leq 2e^{-\frac{N}{4\epsilon\log N}}+\frac{2N}{k} \max_{j\in B}\left\{ \Pr_{v\sim \mu_t}[\textnormal{wt}(vH)=j] \cdot \frac{2^N}{\binom{N}{j}}\right\} \nonumber\\
&+\frac{2^{h(\epsilon)N+5h(\frac{1}{\sqrt{\log N}})N+1}}{k}  \max_{ j\notin B}\left\{ \Pr_{v\sim \mu_t}\big[\textnormal{wt}(vH)=j\big] 2^{2\epsilon N\log|1-\frac{2j}{N}|} \right\},
\end{align}
where $\beta:=\frac{1}{2}\left( 1-2\sqrt{\tilde{\epsilon}(1-\tilde{\epsilon})} \right)$ for $\tilde{\epsilon}:=\epsilon+\frac{1}{\sqrt{\log N}}$, and
$B:=[\beta N,(1-\beta)N]\cap\N$. Our goal will be to bound both the central terms $j\in B$ and the faraway terms $j\notin B$ by using our bounds on the weight distribution of transitive codes. As we've seen in section \ref{prelim}, the dual code $C^\perp$ is a transitive linear code of dimension $N-\textnormal{dim }C $. By Theorem \ref{probtransitive}, we thus have that for all $j\in\{0,1,\dotsc,N\}$,
\begin{align}\label{weightbd}
    \Pr_{v\sim\mu_t}\big[\textnormal{wt}(vH)=j \big]\leq 2^{-(1-h(\frac{j}{N}))(1-\eta) N}.
\end{align}
For any $j\in B$, we then have by Lemma \ref{stirling} that
\begin{align*}
    \Pr_{v\sim\mu_t}\big[\textnormal{wt}(vH)=j \big]\cdot \frac{2^N}{\binom{N}{j}}  &\leq 2^{-(1-h(j/N))(1-\eta) N}\cdot \frac{2^N}{\sqrt{\frac{1}{2 N}}\cdot2^{h(j/N)N}}  \\
    &= \sqrt{2N}\cdot 2^{(1-h(j/N))\eta N}.
\end{align*}
But for $j\in B$ we have $\beta\leq \frac{j}{N}\leq 1-\beta$, so the right-hand side is maximized at $j=\lceil\beta N\rceil$. Applying Lemma \ref{pinsker}, we get
\begin{align}\label{centraltermstransitive}
    \max_{j\in B}\left\{\Pr_{v\sim\mu_t}\big[\textnormal{wt}(vH)=j \big]\cdot \frac{2^N}{\binom{N}{j}}\right\}  &\leq \sqrt{2N}\cdot 2^{(1-h(\beta))\eta N} \nonumber\\
    &\leq \sqrt{2N}\cdot 2^{4\tilde{\epsilon}(1-\tilde{\epsilon})\eta N}.
\end{align}
We now turn to the faraway terms of equation (\ref{thmbdtransitive}). By equation (\ref{weightbd}), we have
\begin{align*}
    \max_{ j\notin B}\left\{ \Pr_{v\sim \mu_t}[\textnormal{wt}(vH)=j]\cdot 2^{2\epsilon N\log|1-\frac{2j}{N}|} \right\}\leq\max_{ \delta< \beta}\left\{  2^{-(1-h(\delta))(1-\eta) N}\cdot 2^{2\epsilon N\log(1-2\delta)} \right\}.
\end{align*}
Note that by definition of $\beta$, any $\delta\in(0,\beta)$ can be written as $\delta=\frac{1-2\sqrt{\alpha\tilde{\epsilon}(1-\tilde{\epsilon})}}{2}$ for some $\alpha> 1$. By Lemma \ref{pinsker}, we can then rewrite our previous expression as
\begin{align*}
    \max_{ j\notin B}\left\{ \Pr_{v\sim \mu_t}[\textnormal{wt}(vH)=j] 2^{2\epsilon N\log|1-\frac{2j}{N}|} \right\}\leq\max_{ \alpha> 1}\left\{  2^{-\frac{2\alpha\tilde{\epsilon}(1-\tilde{\epsilon})}{\ln2}(1-\eta) N} 2^{\epsilon N\log(4\alpha\tilde{\epsilon}(1-\tilde{\epsilon}))} \right\}.
\end{align*}
But for any positive constant $c$, the derivative of $\log (\alpha) - c \alpha$ is $\frac{1}{\alpha\cdot\ln2} - c$, and the second derivative is always negative. Thus, the above expression achieves its maximum when $\alpha = \frac{\epsilon }{2\Tilde{\epsilon}(1-\Tilde{\epsilon})(1-\eta)}$. We then get
\begin{align}\label{farawayterms}
    \max_{ j\notin B}\left\{ \Pr_{v\sim \mu_t}[\textnormal{wt}(vH)=j]\cdot 2^{2\epsilon N\log|1-\frac{2j}{N}|} \right\}&\leq  2^{-\frac{\epsilon N}{\ln2}}\cdot 2^{\epsilon N\log(\frac{2\epsilon}{1-\eta})}\nonumber\\
    &\leq  2^{-h(\epsilon)N}\cdot
    2^{\epsilon N\log(\frac{2}{1-\eta})},
\end{align}
where in the last line we used the inequality $\log(1-x)\geq -\frac{x}{(1-x)\ln2}$ for $x<1$ to get $h(\epsilon)\leq -\epsilon\log(\epsilon)+\frac{\epsilon}{\ln2}$.
We now use equations (\ref{centraltermstransitive}) and (\ref{farawayterms}) to bound the central and faraway terms of (\ref{thmbdtransitive}) respectively. This gives
\begin{align*}
\Pr_{\rho\sim P_\epsilon}[\rho\notin D_k(H\rho^\intercal)]&\leq 2e^{-\frac{N}{4\epsilon\log N}}+\frac{2N}{k}\cdot \sqrt{2N}\cdot 2^{4\tilde{\epsilon}(1-\tilde{\epsilon})\eta N}\\
&\quad+\frac{2^{5h(\frac{1}{\sqrt{\log N}})N+1}}{k}  \cdot 2^{\epsilon N\log(\frac{2}{1-\eta})}\\
&\leq 2e^{-\frac{N}{4\epsilon\log N}}+\frac{2^{5h(\frac{1}{\sqrt{\log N}})N+1}}{k} \cdot (2^{4\epsilon \eta N}+2^{\epsilon N\log(\frac{2}{1-\eta})}).
\end{align*}
We have shown (\ref{eqvlttransitive}), and so we are done.
\end{proof}

\section{List decoding for doubly transitive codes}\label{listrm}
We will now turn to proving our list-decoding bounds for doubly transitive codes. We restate and prove our Theorem \ref{rmclose} below.

\newtheorem*{rmclose}{Theorem \ref{rmclose}}
\begin{rmclose}
Fix any $\epsilon\in(0,\frac{1}{2})$
%, any $N> \max\Big\{\left(\frac{5}{\epsilon}\right)^{20},\left(\frac{5}{\frac{1}{2}-\epsilon}\right)^{20}\Big\} $,
and any $\gamma\leq 1-\log(1+2^{-4\epsilon})$ . Then any doubly  transitive linear code $C\subseteq\F_2^N$ of dimension $\textnormal{dim }C=(1-\gamma) N$ can with high probability list-decode $\epsilon$-errors using a list $T$ of size
$$|T|=2^{h(\epsilon)N-\gamma N+o(N)} .$$
\end{rmclose}

\begin{proof}
We will show that for all $N> 2^{\frac{1}{\epsilon^2(1-\epsilon)^2}+1} $, there exists a function $T$ mapping every $x\in\F_2^N$ to a subset $T(x)\subseteq  C$ of size
         $$|T(x)|=2^{h(\epsilon)N-\gamma N+o(N)},$$
     with the property that for every codeword $c\in  C$ we have
    \begin{align*}
    \Pr_{\rho\sim P_\epsilon}\big[c\notin T(c+\rho)\big]\leq
   3e^{-\frac{N}{4\epsilon\log N}}.
    \end{align*}
Let $H$ denote the parity-check matrix of $C.$ By Lemma \ref{eqvltdecoder}, it is sufficient to show that for any $N> 2^{\frac{1}{\epsilon^2(1-\epsilon)^2}+1} $ and any list size $k>N$, we have
\begin{align}\label{eqvltrm}
\Pr_{\rho\sim P_\epsilon}[\rho\notin D_k(H\rho^\intercal)]
&\leq 2e^{-\frac{N}{4\epsilon\log N}}+\frac{a}{k}\cdot
  2^{h(\epsilon)N-\gamma N}
\end{align}
for some $a=2^{o(N)}.$ Setting the list size $k=a\cdot e^{\frac{N}{4\epsilon\log N}}\cdot
  2^{h(\epsilon)N-\gamma N}$ in equation (\ref{eqvltrm}) will then recover our theorem statement. We thus turn to proving (\ref{eqvltrm}).  We note that $2\lfloor\frac{N}{\sqrt{\log N}}\rfloor+1<\frac{k}{2}$, so there exists some $k^*\in[\frac{k}{2},k]$ satisfying the conditions of Proposition \ref{weightcriterion}. Proposition \ref{weightcriterion} then yields the following bound on the left-hand side of (\ref{eqvltrm}).
\begin{align}\label{thmbdrm}
\Pr_{\rho\sim P_\epsilon}[\rho\notin D_k(H\rho^\intercal)]&\leq 2e^{-\frac{N}{4\epsilon\log N}}+\frac{2N}{k} \max_{j\in B}\left\{ \Pr_{v\sim \mu_t}[\textnormal{wt}(vH)=j] \cdot \frac{2^N}{\binom{N}{j}}\right\} \nonumber\\
&+\frac{2^{h(\epsilon)N+5h(\frac{1}{\sqrt{\log N}})N+1}}{k}  \max_{ j\notin B}\left\{ \Pr_{v\sim \mu_t}\big[\textnormal{wt}(vH)=j\big] 2^{2\epsilon N\log|1-\frac{2j}{N}|} \right\},
\end{align}
where $\beta:=\frac{1}{2}\left( 1-2\sqrt{\tilde{\epsilon}(1-\tilde{\epsilon})} \right)$ for $\tilde{\epsilon}:=\epsilon+\frac{1}{\sqrt{\log N}}$, and
$B:=[\beta N,(1-\beta )N]\cap\N$. %Our goal is to bound every term in these sums by using Samorodnitsky's weight bound for doubly transitive codes.
Our goal will be to bound both the central terms $j\in B$ and the faraway terms $j\notin B$ by using Samorodnitsky's weight distribution bound for doubly transitive codes.
Now by Claim \ref{dual2transitive}, the dual code of a doubly transitive code is itself doubly transitive. Applying  Theorem \ref{previousboundsmall}, we thus get that for all $j\in\{0,1,\dotsc, N\}$,
\begin{align}\label{repeatsamorodnitsky}
    \Pr_{v\sim\mu_t}\big[\textnormal{wt}(vH)=j \big]\leq 2^{-\gamma N+o(N)}\cdot
\left(\frac{1}{2^{1-\gamma}-1}\right)^{\min\{j,N-j\}}.
\end{align}
It then follows that
\begin{align}\label{2transitivecenter}
    \max_{j\in B}\left\{ \Pr_{v\sim\mu_t}\big[\textnormal{wt}(vH)=j \big]\cdot \frac{2^N}{\binom{N}{j}}  \right\}&\leq \max_{\alpha\in[\beta,\frac{1}{2}] }\left\{  2^{-\gamma N-\alpha N \log(2^{1-\gamma}-1)+N-h(\alpha)N+o(N)} \right\}.
\end{align}
We want to bound the expression on the right-hand side by $2^{h(\epsilon)N-\gamma N+o(N)}$.
For this we define the function
\begin{align*}
    f(\alpha):=-\gamma N-\alpha N \log(2^{1-\gamma}-1)+N-h(\alpha)N
\end{align*}
and compute its derivative
\begin{align*}
    \frac{df}{d\alpha}=-N\log(2^{1-\gamma}-1)-N\log\frac{1-\alpha}{\alpha}.
\end{align*}
We note that over the interval $[0,1]$, the second derivative $\frac{d^2f}{d\alpha^2}=\frac{N}{\alpha(1-\alpha)\ln2}$ is positive.  Thus over $[0,1]$, the function $f$ is minimized at the point $\alpha^*$ satisfying $\frac{1-\alpha^*}{\alpha^*}=2^{1-\gamma}-1$ (i.e. $\alpha^*=1-2^{\gamma-1}$), and $f$ is monotone on either side of $\alpha^*$. In particular, over the interval $[\beta,\frac{1}{2}]$ the function $f$ must be maximized at either $\alpha=\beta$ or $\alpha=\frac{1}{2}$. But since $\gamma\leq 1-\log(1+2^{-4\epsilon})$ by our theorem assumption, we have
\begin{align}\label{bdhalf}
    f\big(\frac{1}{2}\big)&\leq-\gamma N+2\epsilon N\nonumber\\&\leq -\gamma N+h(\epsilon)N.
\end{align}

On the other hand we have $\beta=\frac{1-\sqrt{4\epsilon(1-\epsilon)}}{2}-o(1)$
, so in order to show that
%\begin{align}
    %\max_{\alpha\in[\beta,\frac{1}{2}]}\{f(\alpha)\}\leq h(\epsilon)N-\gamma N+o(N),
%\end{align}
\begin{align}\label{goaleqbeta}
    f(\beta)\leq h(\epsilon)N-\gamma N+o(N),
\end{align}
it suffices to show that
\begin{align*}
    -\frac{1-\sqrt{4\epsilon(1-\epsilon)}}{2} \log(2^{1-\gamma}-1)+1-h\Big(\frac{1-\sqrt{4\epsilon(1-\epsilon)}}{2} \Big)-h(\epsilon)\leq 0.
\end{align*}
But the left-hand is an increasing function of $\gamma$, so by our theorem assumption that $\gamma\leq 1-\log(1+2^{-4\epsilon})$, it suffices to show that
\begin{align}\label{technicalfcn}
    2\epsilon(1-\sqrt{4\epsilon(1-\epsilon)})+1-h\Big(\frac{1-\sqrt{4\epsilon(1-\epsilon)}}{2} \Big)-h(\epsilon)\leq 0.
\end{align}
We postpone the proof of this fact to Appendix \ref{atechnical}. Assuming this fact we get equation (\ref{goaleqbeta}), which when combined with (\ref{bdhalf}) and (\ref{2transitivecenter}) gives us
\begin{align}\label{centerterms}
    \max_{j\in B}\left\{ \Pr_{v\sim\mu_t}\big[\textnormal{wt}(vH)=j \big]\cdot \frac{2^N}{\binom{N}{j}}  \right\}&\leq 2^{h(\epsilon)N-\gamma N+o(N)}.
\end{align}
This finishes our analysis of the central terms of equation (\ref{thmbdrm}). For the faraway terms, by (\ref{repeatsamorodnitsky}) we have
\begin{align}\label{intermediatefarterm2}
&\max_{ j\notin B} \left\{\Pr_{v\sim \mu_t}[\textnormal{wt}(vH)=j]\cdot   2^{2\epsilon N\log|1-\frac{2j}{N}|} \right\}\nonumber\\
&\quad\quad\quad\quad\quad\quad\quad\quad\leq  \max_{ j\leq\frac{N}{2}} \left\{  2^{-\gamma N+o(N)} \left(\frac{1}{2^{1-\gamma}-1}\right)^{j}  \cdot 2^{2\epsilon N\log(1-\frac{2j}{N})} \right\}\nonumber\\
&\quad\quad\quad\quad\quad\quad\quad\quad=2^{-\gamma N+o(N)} \max_{ j\leq \frac{N}{2}} \left\{   2^{-j\log(2^{1-\gamma}-1)+2\epsilon N\log(1-\frac{2j}{N})} \right\}.
\end{align}
Now the function
$$g(j):=-j\log(2^{1-\gamma}-1)+2\epsilon N\log(1-\frac{2j}{N})$$
has first derivative
$$\frac{dg}{dj}=-\log(2^{1-\gamma}-1)-\frac{4\epsilon}{\ln2\cdot (1-\frac{2j}{N})},$$
and second derivative
$$\frac{dg^2}{d^2j}=-\frac{8\epsilon}{\ln2\cdot N (1-\frac{2j}{N})^2}<0.$$
Thus $g(j)$ achieves its maximum at $j^*=\frac{N}{2}+\frac{2\epsilon N}{\ln2\log(2^{1-\gamma}-1)}$ and is decreasing over $[j^*,\frac{N}{2}]$. Whenever $1-\gamma\geq\log(1+2^{-\frac{4\epsilon}{\ln2}})$, we have $j^*\leq0$; in that case the argument in equation (\ref{intermediatefarterm2}) is maximized at $j=0$ and we get
\begin{align*}
\max_{ j\notin B} \left\{\Pr_{v\sim \mu_t}[\textnormal{wt}(vH)=j]\cdot   2^{2\epsilon N\log|1-\frac{2j}{N}|} \right\}
&\leq 2^{-\gamma N+o(N)}.
\end{align*}
We now combine this bound for the faraway terms with the bound (\ref{centerterms}) for the central terms to bound the right-hand side of (\ref{thmbdrm}). We get that for all $N> 2^{\frac{1}{\epsilon^2(1-\epsilon)^2}+1} $, we have
\begin{align*}
\Pr_{\rho\sim P_\epsilon}[\rho\notin D_k(H\rho^\intercal)]
&\leq 2e^{-\frac{N}{4\epsilon\log N}}+\frac{2^{o(N)}}{k}\cdot
  2^{h(\epsilon)N-\gamma N}.
\end{align*}
We have shown (\ref{eqvltrm}), so we are done.
\end{proof}

\section*{Acknowledgments}
We thank Alexander Barg, Paul Beame, Noam Elkies, Jan Hazla, Amir Shpilka, Madhu Sudan and Amir Yehudayoff for useful discussions.

\appendix

\section{Weight bounds comparisons}\label{aweight}
In this section, we compare our Theorem \ref{probtransitive} with previously known bounds on the weight distribution of Reed-Muller codes. We will denote by $\mathsf{RM}_q(n,d)$ the  Reed-Muller code over $\F_q$ with $n$ variables and degree $d$. The codewords of $\mathsf{RM}_q(n,d)$ are the evaluation vectors (over all points in $\F_q^n$) of all multivariate polynomials of degree $\leq d$ in $n$ variables. Let $M_q(n, d)$ denote the set of monomials $m=\prod_{i=1}^n x_i^{d_i}$ satisfying
\begin{enumerate}
    \item $d_i<q$ for every $i\in\{1,2,\dotsc, n\}$
    \item $\sum_{i=1}^n d_i\leq d.$
\end{enumerate}
Then the dimension of the corresponding Reed-Muller code is
\begin{align}\label{dimensionrm}
    \textnormal{dim } \mathsf{RM}_q(n,d) = |M_q(n, d)|.
\end{align}

%We recall our Theorem \ref{rmupperprob} below.
We note that when $q=2$, we have $$|M_2(n, d)|=\binom{n}{\leq d}.$$ Throughout this section, we will denote by $\mathcal{D}_q(n,d)$ the uniform distribution over all codewords in $\mathsf{RM}_q(n,d)$, and by $\textnormal{wt}(c)$ the number of non-zero coordinates of $c$. When $q=2$, we will simply write $\mathsf{RM}(n,d)$ and $\mathcal{D}(n,d)$ to mean $\mathsf{RM}_{2}(n,d)$ and $\mathcal{D}_{2}(n,d)$. The following result is an immediate consequence of our Theorem \ref{probtransitive}.

\begin{theorem}\label{rmupperprob}
Consider any finite field $\F_q$. For any non-negative integers $n,d\leq n,$ and any $\alpha\in (0,1)$, the Reed-Muller code $\mathsf{RM}_q(n,d)$ satisfies
 $$ \Pr_{c\sim\mathcal{D}_q(n,d)}\Big[\textnormal{wt}(c) = \alpha N\Big] \leq q^{-(1-h_q(\alpha)) |M_q(n,d)|}.$$
\end{theorem}
\begin{proof}
This follows immediately from Theorem \ref{probtransitive}, Fact \ref{transitive}, and equation (\ref{dimensionrm}).
\end{proof}

\hfill\\
\textbf{Reed-Muller codes over non-prime fields}
\hfill\\
To the best of our knowledge, our Theorem \ref{rmupperprob} is the first weight bound for Reed-Muller codes over non-prime fields.
\hfill\\
\textbf{Reed-Muller codes over odd prime fields}
\hfill\\
For Reed-Muller codes over odd prime fields, the only preexisting weight bound we are aware of is the following result of \cite{beame2020weightodd}:
\begin{theorem}[Corollary 1.2 in \cite{beame2020weightodd}]\label{oddprimes}
For any $0 < \delta < \frac{1}{2}$, there are constants $c_1, c_2>0$ such that for any
odd prime $q$ and for any integers $d,n$ such that $d \leq \delta n$, we have
$$\Pr_{c\sim \mathcal{D}_{q}(n,d)}\Big[\frac{\textnormal{wt}(c)}{N} \leq 1-\frac{1}{q}-q^{-c_1\frac{n}{d}}\Big] \leq q^{-c_2 |M_q(n,d)|}.$$
\end{theorem}
This was a generalization of \cite{ben-eliezer2012weighthalf1}, who proved the same result for Reed-Muller codes over $\F_2$. Theorem \ref{oddprimes} is very strong for small degrees, but gets weaker as the degree increases. When $d$ is linear in $n$ we have $q^{-c_1\frac{n}{d}}=\Theta(1)$, meaning that in this regime Theorem \ref{oddprimes} can only give a nontrivial bound on relative weights $\frac{\textnormal{wt}(c)}{N}$ that are at least a constant away from $1-\frac{1}{q}$. Our Theorem \ref{rmupperprob} gives nontrivial bounds for all relative weights and all degrees.

\hfill\\
\textbf{Reed-Muller codes over $\F_2$}
\hfill\\
We now turn to Reed-Muller codes over $\F_2$, for which more results are known. The same bound as Theorem \ref{oddprimes} was proven over $\F_2$ by \cite{ben-eliezer2012weighthalf1}. For comparison with our Theorem \ref{rmupperprob}, see the discussion above.

In the constant-rate regime (i.e. $d=\frac{n}{2}\pm O(\sqrt{n})$), the strongest known weight bound (for all weights) is due to Samorodnitsky. It follows immediately from Theorem \ref{previousboundsmall}, i.e. from Proposition 1.4 in \cite{samorodnitsky2020weightimproved}.

\begin{theorem}[follows from Proposition 1.4 in \cite{samorodnitsky2020weightimproved}]\label{samorodnitsky1}
 For any $\alpha\in(0,1)$, define $\alpha^{*}:=\min\{\alpha,1-\alpha\}.$  Then for any non-negative integers $n,d\leq n$ and any $\alpha\in (0,1)$, the Reed-Muller code $\mathsf{RM}(n,d)$ satisfies
%$$\Pr_{c\sim\mathcal{D}(n,d)}\Big[\textnormal{wt}(c) = \alpha N\Big]\leq 2^{-\binom{n}{\leq d}+o(N)}
%\left(\frac{1}{2^{1-\frac{\binom{n}{\leq d}}{N}}-1}\right)^{\alpha^*N}.
%$$
$$\Pr_{c\sim\mathcal{D}(n,d)}\Big[\textnormal{wt}(c) = \alpha N\Big]\leq 2^{-\binom{n}{\leq d}+o(N)}
\Big(2^{1-\frac{\binom{n}{\leq d}}{N}}-1\Big)^{-\alpha^*N}.
$$
Moreover, if $\alpha^*\geq 1-2^{\frac{\binom{n}{\leq d}}{N}-1}$,
$$\Pr_{c\sim\mathcal{D}(n,d)}\Big[\textnormal{wt}(c) = \alpha N\Big]\leq 2^{o(N)}\cdot
\frac{\binom{N}{\alpha N}}{2^N}.
$$
\end{theorem}
When the rate of the code is subconstant (i.e. when the degree is away from $\frac{n}{2}$ ), Theorem \ref{samorodnitsky1} does not give strong bounds. An approach that has been fairly successful in this regime is the line of work of \cite{kaufman2012constantdegree,abbe2015RMlowrate,sberlo2020weightbound}.
To our knowledge, the strongest results for these regimes are due to \cite{sberlo2020weightbound}. We start with their bound for lower weights, i.e. for weights in $[0,\frac{N}{4}].$
\begin{theorem}[Theorem 1.1 in \cite{sberlo2020weightbound}]\label{ss1}
For any $j,n,d\in\N$ with $d\leq n$, we have
\begin{align*}
 \Pr_{c\sim\mathcal{D}(n,d)}[\textnormal{wt}(c)\leq 2^{-j}N]&\leq
2^{-\big(  1-17(\frac{j}{1-\frac{d}{n}}+\frac{2-\frac{d}{n}}{(1-\frac{d}{n})^2})(\frac{d}{n})^{j-1} \big)\binom{n}{\leq d} + O(n^4)}.
\end{align*}
\end{theorem}
We claim that for every $d> \frac{n}{34}$, there is some weight threshold $A_d<\frac{1}{4}$ for which our Theorem \ref{rmupperprob} is stronger than Theorem \ref{ss1} for all weights larger than $ A_d N$.
One way to see this is to note that our Theorem \ref{rmupperprob} satisfies
\begin{align*}
\Pr[\textnormal{wt}(c)\leq 2^{-j}\cdot 2^n]&\leq 2^{-\big( 1-h(2^{-j})  \big)\binom{n}{\leq d}}\\
&\leq 2^{-( 1-2j\cdot2^{-j})\binom{n}{\leq d}},
\end{align*}
while the expression in Theorem \ref{ss1} satisfies
$$2^{-\big(  1-17(\frac{j}{1-\frac{d}{n}}+\frac{2-\frac{d}{n}}{(1-\frac{d}{n})^2})(\frac{d}{n})^{j-1} \big)\binom{n}{\leq d} }\geq 2^{-\big( 1-17j(\frac{d}{n})^{j-1} \big)\binom{n}{\leq d} }.$$
Thus our Theorem \ref{rmupperprob} is stronger than Theorem \ref{ss1} whenever
\begin{align}\label{conditioncomparison}
    j\cdot2^{-(j-1)}<17j\cdot(\frac{d}{n})^{j-1}.
\end{align}
This condition is always satisfied when $d\geq \frac{n}{2}$, so in this range our Theorem \ref{rmupperprob} is stronger than Theorem \ref{ss1} for all weights. When $d<\frac{n}{2}$, condition (\ref{conditioncomparison}) is satisfied whenever
$$j<\frac{\log 17}{\log\frac{n}{2d}}+1.$$
For any $\frac{n}{34}<d<\frac{n}{2}$, this gives a nontrivial range.

This concludes our comparison of Theorem \ref{rmupperprob} with Theorem \ref{ss1}, which was the bound of \cite{sberlo2020weightbound} for weights in $[0,\frac{N}{4}].$ We now turn to their bounds for larger weights.
\begin{theorem}[Theorem 1.3 in \cite{sberlo2020weightbound}]\label{ss2}
Let $j,n\in\N$ and let $0<\gamma(n)<\frac{1}{2}-\Omega\left(  \sqrt{\frac{\log n}{n}}\right)$ be a parameter (which may be constant or depend on $n$) such that $\frac{j+\log\frac{1}{1-2\gamma}}{(1-2\gamma)^2}=o(n).$ Then
$$\Pr_{c\sim\mathcal{D}(n,\gamma n)}[\textnormal{wt}(c)\leq \frac{1-2^{-j}}{2}N]\leq 2^{-2^{-c(\gamma,j)}\binom{n}{\leq d}+O(n^4)},$$
where $c(\gamma,j)=O\left( \frac{\gamma^2 j+\gamma\log\frac{1}{1-2\gamma}}{1-2\gamma}+\gamma\right)$.
\end{theorem}
This bound holds when the degree is smaller than $\frac{n}{2}$. For arbitrary degree, \cite{sberlo2020weightbound} gives the following:
\begin{theorem}[Theorem 1.5 in \cite{sberlo2020weightbound}]\label{ss3}
For any $n,d\in\N$ with $d\leq n$ and any $\delta>0$, we have
\begin{align*}
 \Pr_{c\sim\mathcal{D}(n,d)}[\textnormal{wt}(c)\leq \frac{1-\delta}{2}N]&\leq e^{-\frac{\delta^2}{2}\cdot 2^d}.
\end{align*}
\end{theorem}
\hfill\\
We will start by comparing our Theorem \ref{rmupperprob} with Theorem \ref{ss3}. Applying Lemma \ref{pinsker}, we get from Theorem \ref{rmupperprob} that
\begin{align*}
    \Pr_{c\sim\mathcal{D}(n,d)}[\textnormal{wt}(c)\leq \frac{1-\delta}{2}N]&\leq 2^{-(1-h(\frac{1-\delta}{2}))\cdot\binom{n}{\leq d}}\\
    &\leq e^{-\frac{\delta^2}{2}\cdot \binom{n}{\leq d}}.
\end{align*}
Thus our Theorem \ref{rmupperprob} is strictly stronger than Theorem \ref{ss3} for all $d<n$. We will now compare our Theorem \ref{rmupperprob} with Theorem \ref{ss2}. Applying Lemma \ref{pinsker}, we get from Theorem \ref{rmupperprob} that
\begin{align*}
    \Pr_{c\sim\mathcal{D}(n,d)}[\textnormal{wt}(c)\leq \frac{1-2^{-j}}{2}N]&\leq 2^{-(1-h(\frac{1-2^{-j}}{2}))\cdot \binom{n}{\leq d}}\\
    &\leq 2^{-\frac{2^{-2j}}{2\ln2}\cdot \binom{n}{\leq d}}.
\end{align*}
It follows that our Theorem \ref{rmupperprob} is stronger than Theorem \ref{ss2} whenever $2^{-(2j+1)}\geq 2^{-c(\gamma,j)}$, i.e. whenever
\begin{align*}
    2j+1\leq c(\gamma,j).
\end{align*}
But $c(\gamma,j):=O\left( \frac{\gamma^2 }{1-2\gamma}\cdot j+\frac{\gamma\log\frac{1}{1-2\gamma}}{1-2\gamma}+\gamma\right)$, and
$\frac{\gamma^2 }{1-2\gamma}\rightarrow\infty$
as $\gamma\rightarrow 1/2$. Thus there exists some constant $\gamma^*\in (0,\frac{1}{2})$ such that our Theorem \ref{rmupperprob} is stronger than Theorem \ref{ss2} whenever $d>\gamma^* n$. In private correspondence with Amir Shpilka and Ori Sberlo, we learned that $\gamma^*$ can be computed to be $\gamma^*\approx 0.38$.

\section{Proof of corollary \ref{weightunique}}\label{abinomialweight}
Recall that for any $\epsilon\in(0,1)$ we defined $$A_\epsilon:=\{\alpha N\in \N:h(\alpha)>1-h(\epsilon)-N^{-1/5}\},$$
and that for any code $C$ we denote by $\mathcal{D}(C^\perp)$ the uniform distribution over the dual code $C^\perp$. We now restate and prove our Corollary \ref{weightunique}.
\newtheorem*{weightunique}{Corollary \ref{weightunique}}
\begin{weightunique}
%Let $C\subseteq\F_2^N$ be a linear code, and let $\epsilon\in(0,\frac{1}{2})$ be such that $N>\frac{1}{\epsilon^4(\frac{1}{2}-\epsilon)^4} $.
Let $\epsilon\in(0,\frac{1}{2})$ be arbitrary, and let $C\subseteq\F_2^N$ be a linear code.
Suppose that for every $j\in A_\epsilon$ we have
\begin{align*}
    \Pr_{y\sim\mathcal{D}(C^\perp)}\big[\textnormal{wt}(y)=j\big]\leq \big( 1+o(N^{-1})  \big)\frac{\binom{N}{j}}{2^N},
\end{align*}
and suppose that
\begin{align*}
    \Pr_{y\sim\mathcal{D}(C^\perp)}\big[\textnormal{wt}(y)\notin A_\epsilon\big]\leq 2^{N^{\frac{3}{4}}}\cdot \frac{\sum_{i\notin A_\epsilon}\binom{N}{i}}{2^N}.
\end{align*}
Then $C$ is resilient to $\epsilon$-errors.
\end{weightunique}

\begin{proof}
From Theorem \ref{fouriercriterionunique}, we know that whenever $N>\frac{1}{\epsilon^4(\frac{1}{2}-\epsilon)^4} $, there exists some decoder $d:\F_2^N\rightarrow C$ such that for all $c\in C$,
\begin{align}\label{summation}
    \Pr_{\rho\sim P_\epsilon}[ d(c+\rho)\neq c]&\leq 2e^{-\frac{\sqrt{N}}{3\epsilon}}\nonumber\\
&\quad+N \mathop{\max}_{\substack{S\subseteq [\epsilon N\pm N^{3/4}]\cap\N \\
    1\leq|S|\leq 2}} \Big\{ \frac{1}{\binom{N}{S}}\sum_{j=0}^N \Pr_{y\sim C^\perp}\big[\textnormal{wt}(y)=j\big] K_S(j)^2-1 \Big\}.
\end{align}
Let $\nu\in(0,\frac{1}{2})$ be such that $ h(\nu)=1-h(\epsilon)-N^{-1/5}$, and note that we have
$$A_\epsilon=\{\lceil \nu N\rceil,\lceil \nu N\rceil+1,\dotsc,\lfloor(1-\nu)N\rfloor\}.$$
We will start by bounding the central terms $j\in A_\epsilon$ in equation (\ref{summation}). Applying Proposition \ref{IFourier} and the first condition in our theorem statement, we immediately get that for any $S\subseteq \{0,1,\dotsc,N\}$,
\begin{align}\label{center}
    \frac{1}{\binom{N}{S}}\sum_{j\in A_\epsilon} \Pr_{y\sim C^\perp}\big[\textnormal{wt}(y)=j\big] K_S(j)^2&\leq 1+o\big(\frac{1}{N}\big).
\end{align}
We now turn to the faraway terms $j\notin A_\epsilon$. For these, we note that for any non-negative integers $j,s\leq N$ we have
\begin{align*}
    |K_s(j)|&=\left|\sum_{t=0}^s(-1)^t\binom{j}{t}\binom{N-j}{s-t}\right| \\
    &\leq\sum_{t=0}^s\binom{j}{t}\binom{N-j}{s-t} \\ &=\binom{N}{s},
\end{align*}
where we used the convention that $\binom{a}{b}=0$ when $a<b$. For any $S\subseteq \{0,1,\dotsc,N\}$, we can then bound the faraway terms $j\notin A_\epsilon$ of equation (\ref{summation}) by
\begin{align*}
\frac{1}{\binom{N}{S}}\sum_{j\notin A_\epsilon} \Pr_{y\sim C^\perp}\big[\textnormal{wt}(y)=j\big] K_S(j)^2&\leq\binom{N}{S}\Pr_{y\sim C^\perp}\big[ \textnormal{wt}(y)\notin A_\epsilon\big].
\end{align*}
Applying the second condition in our theorem statement in combination with Lemma \ref{stirling} and the subadditivity of entropy (Lemma \ref{entropysubadditive}), we get
\begin{align*}
&\mathop{\max}_{\substack{S\subseteq [\epsilon N\pm N^{3/4}]\cap\N \\
    1\leq|S|\leq 2}} \Big\{ \frac{1}{\binom{N}{S}}\sum_{j\notin A_\epsilon} \Pr_{y\sim C^\perp}\big[\textnormal{wt}(y)=j\big] K_S(j)^2\Big\}\\
&\quad\quad\quad\quad\quad\quad\quad\quad\leq 2\binom{N}{\lfloor\epsilon N+N^{3/4}\rfloor}\cdot 2\cdot2^{-h(\epsilon)N-N^{4/5}+N^{3/4}}\nonumber\\
&\quad\quad\quad\quad\quad\quad\quad\quad\leq 4\cdot 2^{h(\epsilon)N+h(N^{-1/4})N}\cdot 2^{-h(\epsilon)N-N^{4/5}+N^{3/4}}\nonumber\\
&\quad\quad\quad\quad\quad\quad\quad\quad\leq o(\frac{1}{N}).
\end{align*}
Combining this bound for the faraway terms with our bound (\ref{center}) for the central terms, we bound equation (\ref{summation}) by
\begin{align*}
\Pr_{\rho\sim P_\epsilon}[ d(c+\rho)\neq c)]&\leq    2e^{-\frac{\sqrt{N}}{3\epsilon}}+N\cdot o\big(\frac{1}{N}\big)\\
&\leq o(1).
\end{align*}

\end{proof}

\section{Lower bounds on list decoding}\label{alistcapacity}
In this section, we prove the result mentioned in equation (\ref{lwboundlist}), section \ref{intro}.
\begin{claim}
Let $\epsilon\in(0,\frac{1}{2})$ be arbitrary, and consider any $N>\frac{100}{\epsilon^2}$. Suppose a code $C\subseteq\F_2^N$ and a decoder $d_k:\F_2^N\rightarrow C^{\otimes k}$ satisfy
$$\mathop{\Pr}_{\substack{\rho\sim P_\epsilon\\c\sim \mathcal{D(C)}}}[c\in d_k(c+\rho)]\geq \frac{3}{4},$$
for $P_\epsilon$ the $\epsilon$-noisy distribution and $\mathcal{D(C)}$ the uniform distribution on $C$. Then we must have $$k\geq |C|\cdot 2^{-(1-h(\epsilon))N}\cdot \frac{2^{-h(\epsilon)N^{3/4}}}{8}.$$
\end{claim}
\begin{proof}
We will first show that in order for the decoder $d_k$ to succeed with high probability, there must be many codewords $c\in C$ for which $$|\{ x\in\F_2^N:c\in d_k(x) \}|\gtrsim 2^{h(\epsilon)N}.$$
Intuitively, this is because the sphere of radius $\epsilon N$ around any codeword $c$ contains $\approx2^{h(\epsilon)N}$ points (and for any transmitted codeword $c$, with high probability the received message $m$ will satisfy $\textnormal{wt}(m+c)\approx \epsilon N$). We will then simply double-count the number of pairs $(x, c)$ for which $c\in d_k(x)$. On the one hand, there are $2^N\cdot k$ such pairs, since every received message is mapped to $k$ codewords; on the other hand, there must be at least about $|C|\cdot 2^{h(\epsilon)N}$ pairs, since as we've just argued most codewords in $C$ need to be matched to at least $\approx 2^{h(\epsilon)N}$ points. It follows that we must have
$$k\gtrsim |C|\cdot \frac{2^{h(\epsilon)N}}{2^N}.$$

Formally, we first note that the theorem condition implies that at least $\frac{|C|}{2}$ codewords $c\in C$ must satisfy
\begin{align}\label{rewritecondition}
    \Pr_{\rho\sim P_\epsilon}[c\in d_k(c+\rho)]\geq \frac{1}{2}.
\end{align}
Fix any such $c$. Now from Chernoff's bound (i.e Lemma \ref{chernoff}), we have for $N>\frac{100}{\epsilon^2}$ that
\begin{align*}
    \Pr_{\rho\sim P_\epsilon}\big[\textnormal{wt}(\rho)\leq \epsilon N- \epsilon N^{3/4}\big]&\leq 2e^{-\frac{10}{3}}\\
    &\leq\frac{1}{4}.
\end{align*}
In order for $c$ to satisfy $c\in d_k(c+\rho)$ with probability at least $\frac{1}{2}$, there must then be a subset $S_c\subseteq\{x\in\F_2^N:\textnormal{wt}(c+x)\geq\epsilon N-\epsilon N^{3/4}\} $ satisfying both
\begin{align}\label{eq1}
    x\in S_c\implies c\in d_k(x)
\end{align}
and
\begin{align}\label{eq2}
    \Pr_{\rho\sim P_\epsilon}\big[\rho\in S_c\big]\geq \frac{1}{4}.
\end{align}
But every element $x\in S_c$ satisfies $\textnormal{wt}(c+x)\geq \epsilon N- \epsilon N^{3/4}$, so every $x\in S_c$ satisfies
\begin{align}\label{eq3}
\Pr_{\rho\sim P_\epsilon}\big[\rho=c+x\big]&\leq \epsilon^{\epsilon N-\epsilon N^{3/4}}(1-\epsilon)^{(1-\epsilon)N+\epsilon N^{3/4}}\nonumber\\
&\leq2^{-(1-N^{-1/4})h(\epsilon)N}
\end{align}
Equations (\ref{eq2}) and (\ref{eq3}) imply that any $c\in C$ that can be list-decoded by $d_k$ with probability $\geq \frac{1}{2}$ must satisfy $|S_c|\geq \frac{2^{(1-N^{-1/4})h(\epsilon)N}}{4}$. It then follows from (\ref{eq1}) that any such $c$ must satisfy
$$\big|\{x\in\F_2^N:c\in d_k(x)\}   \big|\geq \frac{2^{(1-N^{-1/4})h(\epsilon)N}}{4}.$$
By double counting, we get
\begin{align*}
    2^N\cdot k&=\sum_{c\in C}\big|\{x\in\F_2^N:c\in d_k(x)\}   \big|\\
    &\geq \frac{|C|}{2}\cdot \frac{2^{(1-N^{-1/4})h(\epsilon)N}}{4}\\
    &=\frac{|C|}{8}\cdot 2^{h(\epsilon)N-h(\epsilon)N^{3/4}}.
\end{align*}
The result then follows from rearranging terms.
\end{proof}

\section{Other proofs for sections \ref{intro}, \ref{prelim} and \ref{outline}}

\subsection{Explicit bounds from Theorem \ref{thmtransitivelist}}\label{calculationsthm}
In this section, we prove the result we mentioned in equation (\ref{egtransitive}).
\begin{claim}
Fix any $\epsilon\in(0,\frac{1}{2})$ and $N> 2^{\frac{1}{\epsilon^2(1-\epsilon)^2}+1} $. Then any  transitive linear code $C\subseteq\F_2^N$ of dimension $\textnormal{dim }C=(1-\frac{4\epsilon}{e})N$
can with high probability list-decode $\epsilon$-errors using a list $T$ of size
\begin{align*}
|T|= 2^{(h(\epsilon) -\epsilon +\frac{\epsilon^2}{\ln2} )N+o(N)}+2^{4\epsilon N+o(N)}.
\end{align*}
\end{claim}
\begin{proof}
From Theorem \ref{thmtransitivelist}, we know that $C$ can with high probability list-decode $\epsilon$-errors using a list $T$ of size
\begin{align*}
|T|=&2^{\epsilon N\log(\frac{2e}{4\epsilon})+o(N)}+2^{4\epsilon N+o(N)}\\
&=2^{\epsilon N\log(\frac{1}{\epsilon})+\epsilon N\log e-\epsilon N+o(N)}+2^{4\epsilon N+o(N)}\\
&=2^{\epsilon N\log(\frac{1}{\epsilon})+(1-\epsilon)N\frac{\epsilon}{\ln2}-\epsilon N+\frac{\epsilon^2}{\ln2}N+o(N)}+2^{4\epsilon N+o(N)}\\
&\leq 2^{(h(\epsilon) -\epsilon +\frac{\epsilon^2}{\ln2} )N+o(N)}+2^{4\epsilon N+o(N)},
\end{align*}
where in the last line we used the inequality $\log(1-x)\leq -\frac{x}{\ln2}$ for all $x$ to get $h(\epsilon)\geq \epsilon\log\frac{1}{\epsilon}+(1-\epsilon)\frac{\epsilon}{\ln2}$.
\end{proof}

\subsection{Duals of transitive codes - proof of claims \ref{dualtransitive} and \ref{dual2transitive}}\label{adualtransitive}
We show that the dual of a transitive code is itself transitive.
\newtheorem*{dualtransitive}{Claim \ref{dualtransitive}}
\begin{dualtransitive}
The dual code $C^\perp$ of a transitive code $C\subseteq\F_2^N$ is transitive.
\end{dualtransitive}

\begin{proof}
Let $i,j\in [N]$ be arbitrary. Since $C$ is transitive, we know there exists a permutation $\pi :[N]\rightarrow [N]$ such that $\pi(j)=i$ and for any $c=(c_1,c_2,\dotsc,c_N)\in C$, we have $c_{\pi}:=(c_{\pi(1)},c_{\pi(2)},\dotsc,c_{\pi(N)}) \in C$ . Clearly $\pi^{-1}$ satisfies $\pi^{-1}(i)=j$, and we claim that it also satisfies that $v_{\pi^{-1}} \in C^\perp$ for all $v\in C^\perp$. For this we note that since $c_\pi\in C$ for every $c\in C$, we have by definition that every $v\in C^\perp$ satisfies
$$\sum_k v_k c_{\pi(k)} = 0 \textnormal{ for all }c\in C.$$
We thus have
\begin{align*}
    v\in C^\perp &\implies \sum_k v_k c_{\pi(k)} = 0 \textnormal{ for all }c\in C\\
    &\implies \sum_k v_{\pi^{-1}(k)} c_k = 0 \textnormal{ for all }c\in C \\
    &\implies v_{\pi^{-1}}\in C^\perp.
\end{align*}
\end{proof}

\newtheorem*{dual2transitive}{Claim \ref{dual2transitive}}
\begin{dual2transitive}
The dual code $C^\perp$ of a doubly transitive code $C\subseteq\F_2^N$ is doubly transitive.
\end{dual2transitive}

\begin{proof}
Let $i,j,k,l\in [N]$ be such that $i\neq k$ and $j\neq l.$ Since $C$ is doubly transitive, we know there exists a permutation $\pi :[N]\rightarrow [N]$ such that $\pi(j)=i$, $\pi(l)=k$,  and for any $c=(c_1,c_2,\dotsc,c_N)\in C$, we have $c_{\pi}:=(c_{\pi(1)},c_{\pi(2)},\dotsc,c_{\pi(N)}) \in C$ . Clearly $\pi^{-1}$ satisfies $\pi^{-1}(i)=j$ and $\pi^{-1}(k)=l$, and we claim that it also satisfies that $v_{\pi^{-1}} \in C^\perp$ for all $v\in C^\perp$. For this we note that since $c_\pi\in C$ for every $c\in C$, we have by definition that every $v\in C^\perp$ satisfies
$$\sum_{t=1}^N v_t c_{\pi(t)} = 0 \textnormal{ for all }c\in C.$$
We thus have
\begin{align*}
    v\in C^\perp &\implies \sum_t v_t c_{\pi(t)} = 0 \textnormal{ for all }c\in C\\
    &\implies \sum_t v_{\pi^{-1}(t)} c_t = 0 \textnormal{ for all }c\in C \\
    &\implies v_{\pi^{-1}}\in C^\perp.
\end{align*}
\end{proof}

\subsection{On known list-decoding bounds for doubly transitive codes}\label{acomparelist}
We recall the known list-decoding bound for doubly transitive codes (see equation (\ref{previouslistresult}) in section \ref{intro}):
\begin{align*}
    |T|=2^{\epsilon N\log\frac{4\epsilon(1-\epsilon)}{(2^\gamma-1)^2}+o(N)},
\end{align*}
where $1-\gamma\in(0,1)$ is the rate of the code. We claim that for constant $\gamma$, this bound never achieves the information-theoretic $2^{h(\epsilon)N-\gamma N+o(N)}$.
\begin{claim}
For any $\epsilon\in(0,\frac{1}{2})$ and any $\gamma\in (0,1)$, we have
\begin{align*}
    \epsilon \log\frac{4\epsilon(1-\epsilon)}{(2^\gamma-1)^2}>h(\epsilon)-\gamma.
\end{align*}
\end{claim}
\begin{proof}
Since $2^x< 1+x$ for all $x\in(0,1)$, it will be sufficient to show that
\begin{align*}
    \epsilon \log\frac{4\epsilon(1-\epsilon)}{\gamma^2}\geq h(\epsilon)-\gamma.
\end{align*}
We will thus show that for any $\epsilon\in(0,\frac{1}{2})$ and any $c=\frac{\gamma}{\epsilon}<\frac{1}{\epsilon}$, we have
\begin{align*}
    \epsilon \log\frac{4\epsilon(1-\epsilon)}{(c\epsilon)^{2}}\geq h(\epsilon)-c\epsilon,
\end{align*}
i.e. that
\begin{align}\label{goallistbnd}
    f(\epsilon,c):=\log(1-\epsilon)+2\epsilon-2\epsilon\log c +c\epsilon\geq0.
\end{align}
We first fix some $\epsilon\in(0,\frac{1}{2})$ and compute the $c$ minimizing $f(\epsilon,c)$. Note that
\begin{align*}
    \frac{\partial}{\partial c} f(\epsilon,c)&=-\frac{2\epsilon}{c\ln2} +\epsilon
\end{align*}
and
\begin{align*}
    \frac{\partial^2}{\partial c^2} f(\epsilon,c)&=\frac{2\epsilon}{c^2\ln2}>0,
\end{align*}
so $f(\epsilon,c)$ is minimized at $c=\frac{2}{\ln2}$ and decreasing over $c\in[0,\frac{2}{\ln2}]$. We thus have
\begin{align}\label{2casesforc}
    \min_{c\leq\frac{1}{\epsilon}} f(\epsilon,c)=\begin{cases}
f(\epsilon,\frac{2}{\ln2}) & \text{if $\epsilon\leq\frac{\ln2}{2}$,}\\
f(\epsilon,\frac{1}{\epsilon}) & \text{otherwise.}
\end{cases}
\end{align}
We deal with each case separately. For the case $\epsilon\leq\frac{\ln2}{2}$, we want to show that
\begin{align*}
    f(\epsilon,\frac{2}{\ln2})=\log(1-\epsilon)+2\epsilon\log (\ln2) +\frac{2\epsilon}{\ln2}\geq 0.
\end{align*}
The first derivative is
\begin{align*}
    \frac{\partial}{\partial \epsilon}f(\epsilon,\frac{2}{\ln2})=-\frac{1}{(1-\epsilon)\ln2}+2\log(\ln2)+\frac{2}{\ln2},
\end{align*}
and the second derivative is
\begin{align*}
    \frac{\partial^2}{\partial \epsilon^2}f(\epsilon,\frac{2}{\ln2})=-\frac{1}{(1-\epsilon)^2\ln2}<0.
\end{align*}
Thus the function $f(\epsilon,\frac{2}{\ln2})$ is maximized at $\epsilon^*=1-\frac{1}{(2\log(\ln2)+\frac{2}{\ln2})\ln2}\approx 0.21,$ and monotone on each side of $\epsilon^*$. In particular, we know that over the interval $[0,\frac{\ln2}{2}]$ the function $f(\epsilon,\frac{2}{\ln2})$ achieves its minimum at either $\epsilon=0$ or $\epsilon=\frac{\ln2}{2}$. But $f(0,\frac{2}{\ln2})=0<f(\frac{\ln2}{2},\frac{2}{\ln2})$, so we indeed have that
$$f(\epsilon,\frac{2}{\ln2})\geq 0$$
for all $0\leq\epsilon\leq\frac{\ln2}{2}$. This deals with the first case of (\ref{2casesforc}). For the second case of (\ref{2casesforc}), we want to show that for all $\epsilon\in(0,\frac{1}{2})$ we have
$$f(\epsilon,\frac{1}{\epsilon})=\log(1-\epsilon)+2\epsilon+2\epsilon\log\epsilon+1\geq0.$$
But
\begin{align*}
\frac{\partial }{\partial \epsilon}f(\epsilon,\frac{1}{\epsilon})&=-\frac{1}{(1-\epsilon)\ln2}+2-2\log(\frac{1}{\epsilon})+\frac{2}{\ln2}
\end{align*}
is maximized at $\epsilon=\frac{1}{2}$, since
$\frac{\partial^2 }{\partial \epsilon^2}f(\epsilon,\frac{1}{\epsilon})=\frac{1}{\ln2}(\frac{2}{\epsilon}-\frac{1}{(1-\epsilon)^2})$ and $2(1-\epsilon)^2\geq \frac{1}{2}\geq \epsilon$ for $\epsilon\in(0,\frac{1}{2})$. It then follows that for $\epsilon\in(0,\frac{1}{2})$, we have
\begin{align*}
\frac{\partial }{\partial \epsilon}f(\epsilon,\frac{1}{\epsilon})&\leq -\frac{1}{(1-\frac{1}{2})\ln2}+2-2\log(2)+\frac{2}{\ln2}\\
&=0,
\end{align*}
and so the function $f(\epsilon,\frac{1}{\epsilon})$ is decreasing in $\epsilon$. Since $f(\frac{1}{2},2)=0$, we indeed have $f(\epsilon,\frac{1}{\epsilon})\geq 0$ for all $\epsilon\in(0,\frac{1}{2})$.
\end{proof}

\subsection{A version of Pinsker's inequality - Proof of Lemma \ref{pinsker}}\label{apinsker}

\newtheorem*{pinsker}{Lemma \ref{pinsker}}
\begin{pinsker}
For any $\mu\in(0,1)$, we have
$$1-h\left(\frac{1-\mu}{2}\right)=\frac{1}{2\ln2}\sum_{i=1}^\infty \frac{\mu^{2i}}{i(2i-1)},$$
and thus
$$\frac{\mu^2}{2\ln2}\leq1-h\left(\frac{1-\mu}{2}\right)\leq \mu^2.$$
\end{pinsker}

\begin{proof}
\begin{align*}
    1-h(\frac{1-\mu}{2})&=1+\frac{1-\mu}{2}\log\left( \frac{1-\mu}{2} \right)+\frac{1+\mu}{2}\log\left( \frac{1+\mu}{2} \right)\\
    &=\frac{1-\mu}{2}\log\left( 1-\mu\right)+\frac{1+\mu}{2}\log\left( 1+\mu\right)\\
    &=\frac{1}{2\ln2}\left[ -(1-\mu)\sum_{i=1}^\infty\frac{\mu^i}{i}-(1+\mu)\sum_{i=1}^\infty(-1)^i\frac{\mu^i}{i}  \right]\\
    &=\frac{1}{2\ln2}\left[ 2\mu\sum_{i=1}^\infty\frac{\mu^{2i-1}}{2i-1}-2\sum_{i=1}^\infty\frac{\mu^{2i}}{2i}  \right]\\
    &=\frac{1}{\ln2}\sum_{i=1}^\infty\mu^{2i}\left( \frac{1}{2i-1}-\frac{1}{2i}  \right)\\
    &=\frac{1}{2\ln2}\sum_{i=1}^\infty \frac{\mu^{2i}}{i(2i-1)}
\end{align*}
Thus $1-h(\frac{1-\mu}{2})\geq\frac{\mu^2}{2\ln2}$ and $1-h(\frac{1-\mu}{2})\leq \frac{1}{2\ln2} \sum_{i=1}^\infty \frac{\mu^2}{i(2i-1)}=\frac{1}{2\ln2}\cdot2\ln2\cdot \mu^2=\mu^2$.
\end{proof}

\subsection{Proof of (\ref{technicalfcn})}\label{atechnical}
\begin{claim}
    For any $\epsilon\in[0,\frac{1}{2}]$, we have
    \begin{align*}
        h\Big(\frac{1-\sqrt{4\epsilon(1-\epsilon)}}{2} \Big)+h(\epsilon)\geq 1+2\epsilon(1-\sqrt{4\epsilon(1-\epsilon)}).
    \end{align*}

\end{claim}
\begin{proof}
Writing the Taylor expansion of
$h$ as in the proof of Lemma \ref{pinsker}, we have
\begin{align*}
    h\Big(\frac{1-\sqrt{4\epsilon(1-\epsilon)}}{2} \Big)+h(\epsilon)&=2-\frac{1}{2\ln2}\sum_{i=1}^\infty \frac{(4\epsilon(1-\epsilon))^{i}+(1-2\epsilon)^{2i}}{i(2i-1)}.
\end{align*}
But $\sum_{i=1}^\infty\frac{1}{i(2i-1)}=2\ln2$, so our previous expression can be rewritten as
\begin{align*}
    h\Big(\frac{1-\sqrt{4\epsilon(1-\epsilon)}}{2} \Big)+h(\epsilon)&=1+\frac{1}{2\ln2}\sum_{i=1}^\infty \frac{1-(4\epsilon(1-\epsilon))^{i}-(1-4\epsilon(1-\epsilon))^{i}}{i(2i-1)}\\
    %&=1+\frac{1}{2\ln2}\sum_{i=2}^\infty \frac{1-(4\epsilon(1-\epsilon))^{i}-(1-4\epsilon(1-\epsilon))^{i}}{i(2i-1)}\\
    &\geq 1+\frac{1}{2\ln2}\sum_{i=2}^\infty \frac{1-(4\epsilon(1-\epsilon))^{2}-(1-4\epsilon(1-\epsilon))^{2}}{i(2i-1)},
\end{align*}
where in the second line we used the fact that the term $i=1$ in the summation is $0.$
We will now need the following inequality:
\begin{align}\label{goalineqln2}
    1-(4\epsilon(1-\epsilon))^{2}-(1-4\epsilon(1-\epsilon))^{2}\geq \frac{4\ln2\cdot\epsilon(1-\sqrt{4\epsilon(1-\epsilon)})}{2\ln2-1}.
\end{align}
Once we establish (\ref{goalineqln2}), our claim follows from bounding our previous inequality by
\begin{align*}
    h\Big(\frac{1-\sqrt{4\epsilon(1-\epsilon)}}{2} \Big)+h(\epsilon)
    &\geq 1+\frac{1}{2\ln2}\cdot\frac{4\ln2\cdot\epsilon(1-\sqrt{4\epsilon(1-\epsilon)})}{2\ln2-1}\sum_{i=2}^\infty \frac{1}{i(2i-1)}\\
    &=1+\frac{2\epsilon(1-\sqrt{4\epsilon(1-\epsilon)})}{2\ln2-1}\Big(\sum_{i=1}^\infty \frac{1}{i(2i-1)}-1\Big)\\
    &=1+2\epsilon(1-\sqrt{4\epsilon(1-\epsilon)})
\end{align*}
It thus only remains to prove (\ref{goalineqln2}). For this, we note that the right-hand side of (\ref{goalineqln2}) can be bounded by
\begin{align*}
    \frac{4\ln2\cdot\epsilon(1-\sqrt{4\epsilon(1-\epsilon)})}{2\ln2-1}\leq 8\epsilon(1-\sqrt{4\epsilon(1-\epsilon})),
\end{align*}
while the left-hand side of (\ref{goalineqln2}) expands to
\begin{align*}
    1-(4\epsilon(1-\epsilon))^{2}-(1-4\epsilon(1-\epsilon))^{2}=8\epsilon-40\epsilon^2+64\epsilon^3-32\epsilon^4.
\end{align*}
Thus it is sufficient to show that
\begin{align*}
    5\epsilon-8\epsilon^2+4\epsilon^3\leq \sqrt{4\epsilon(1-\epsilon)},
\end{align*}
or equivalently (squaring both sides and dividing by $\epsilon$) that the function
\begin{align*}
    g(\epsilon):=16\epsilon^5-64\epsilon^4+104\epsilon^3-80\epsilon^2+29\epsilon-4
\end{align*}
satisfies
\begin{align}\label{goalfunctiong}
    g(\epsilon)\leq 0
\end{align}
for all $\epsilon\in[0,\frac{1}{2}].$ But the derivative of $g$ is
\begin{align*}
    \frac{dg}{d\epsilon}&=80\epsilon^4-256\epsilon^3+312\epsilon^2-160\epsilon+29\\
    &=(1-2\epsilon)^2(20\epsilon^2-44\epsilon+29),
\end{align*}
and the polynomial $20\epsilon^2-44\epsilon+29$ has the two complex roots $\frac{11\pm 2\sqrt{6}\cdot i}{10}$. Thus over the interval $[0,\frac{1}{2}]$, the function $g(\epsilon)$ must be maximized at either $\epsilon=0$ or $\epsilon=\frac{1}{2}$. Since $g(0)=-4$ and $g(\frac{1}{2})=0$, we have $$g(\epsilon)\leq 0$$
for all $\epsilon\in[0,\frac{1}{2}].$ We have thus shown (\ref{goalfunctiong}), and we are done.
\end{proof}

\end{document}